\documentclass[journal,12pt,onecolumn,draftclsnofoot,]{IEEEtran}
\IEEEoverridecommandlockouts
\usepackage[noadjust]{cite}
\usepackage{amsmath,amssymb,amsfonts}
\usepackage{graphicx}
\usepackage{mathtools}
\usepackage{bbm}
\usepackage{float}
\usepackage{textcomp}
\usepackage{xcolor}
\usepackage{comment}
\usepackage{amsthm}
\usepackage{subfig}
\usepackage[ruled,vlined]{algorithm2e}
\usepackage[noend]{algpseudocode}
\allowdisplaybreaks

\newcounter{relctr} 
\everydisplay\expandafter{\the\everydisplay\setcounter{relctr}{0}} 

\newcommand\labelrel[2]{%
  \begingroup
    \refstepcounter{relctr}%
    \stackrel{\textnormal{(\alph{relctr})}}{\mathstrut{#1}}%
    \originallabel{#2}%
  \endgroup
}
\newtheorem{theorem}{Theorem}
\newtheorem{lemma}{Lemma}
\newtheorem{definition}{Definition}
\makeatletter
\DeclareMathOperator*{\argmax}{argmax}

\newcommand{\pushright}[1]{\ifmeasuring@#1\else\omit\hfill$\displaystyle#1$\fi\ignorespaces}
\makeatother
\def\BibTeX{{\rm B\kern-.05em{\sc i\kern-.025em b}\kern-.08em
    T\kern-.1667em\lower.7ex\hbox{E}\kern-.125emX}}
\DeclarePairedDelimiter\ceil{\lceil}{\rceil}
\DeclarePairedDelimiter\floor{\lfloor}{\rfloor}    
\AtBeginDocument{\let\originallabel\label}

\begin{document}

\title{On Minimizing Channel-Aware Age of Information in a Multi-Sensor Setting }

\author{
\IEEEauthorblockN{Bhishma Dedhia\IEEEauthorrefmark{1}\textsuperscript{\textsection}, Sharayu Moharir\IEEEauthorrefmark{2}
}

\IEEEauthorblockA{\IEEEauthorrefmark{1}Department of Electrical Engineering, Princeton University
}

\IEEEauthorblockA{\IEEEauthorrefmark{2}Department of Electrical Engineering, Indian Institute of Technology Bombay
}
}

\maketitle
\begingroup\renewcommand\thefootnote{\textsection}
\footnotetext{Work done at Indian Institute of Technology Bombay}
\endgroup
\nocite{*}

\begin{abstract}
We propose a variant of the Age of Information (AoI) metric called Channel-Aware Age of Information (CA-AoI). Unlike AoI, CA-AoI takes into account the channel conditions between the source and the intended destination to compute the “age” of the recent most update received by the destination. This new metric ensures that the resource allocation is not heavily tilted towards the sources\footnote{We use source and sensor interchangeably.} with poor channel conditions.

We design scheduling policies for multi-sensor systems in which sensors report their measurements to a central monitoring station via a shared unreliable communication channel with the goal of minimizing the time-average of the weighted sum of CA-AoIs \footnote{A preliminary version of this work appeared in the RAWNET Workshop, WiOpt 2020 \cite{9155332}}. We initially derive universal lower bounds for the freshness objective. We show that the scheduling problem is indexable and derive low complexity Whittle index based scheduling policies. We also design stationary randomized scheduling algorithms and give optimization procedures to find the optimal parameters of the policy. Via simulations, we show that our proposed policies surpass the greedy policy in several settings. Moreover the Whittle Index based scheduling policies outperform other policies in all the settings considered. 
\end{abstract}

\section{Introduction}
With the advent of Internet of Things (IoT), multi-sensor systems have become increasingly ubiquitous. They are used to make control decisions in systems ranging from smart homes to vehicular networks. Sensors in these systems measure time varying quantities and communicate with a central monitoring station. It is crucial that the monitoring station receives timely updates from the sensors to make critical decisions. 

In this work, we consider \textit{n} sensors that measure real time updates and update a monitoring station through a shared unreliable communication channel. We divide time into slots and model the channels as ON-OFF channels that have independent realizations across sensors and time slots. At most one sensor is allowed to be scheduled in each time slot under the constraints of the scheduling problem. The goal is to design a scheduling policy that optimizes a freshness objective and hence reduces the staleness in information.

\textit{Age of Information} (AoI) has emerged as a popular metric to measure freshness of information at a destination and is defined as the time elapsed since the last update from a source. It grows linearly with time and resets to zero when an update is received from a source. AoI minimization has been studied extensively for scheduling applications in \cite{7492912}-\cite{8514816}. 
In a multi-sensor setting with a shared channel, the AoI of a sensor increases due to two reasons. The first is when an attempted update by a sensor fails due to poor channel conditions and the second is when the shared channel is used by some other sensor. The original AoI metric does not distinguish between these two events. It is evident that the first event is fundamental and is caused by the physics of the problem and the second event is a resource allocation issue and therefore, controllable by the algorithm designer. Given this, we propose a variant of the AoI which distinguishes between these two events by imposing a penalty by increasing the ``age" of the information of a sensor at the monitoring station only in the latter case. We refer to this metric as the \textit{Channel-Aware Age of Information} (CA-AoI) and by definition, the CA-AoI of a sensor is a measure of the number of missed opportunities to send an update since its recent most successful update.

 The goal of our work is to design scheduling policies to minimize the time-average of the weighted sum of CA-AoIs of the sensors. Another important issue with AoI is that it grows to infinity as the channel conditions of a sensor become poor irrespective of the scheduling policy employed. The following example highlights the difference in schedules designed to minimize the time-average of the sum of the AoIs vs. the CA-AoIs. Consider a two sensor system with equal weights such that the first sensor has near reliable channel conditions and the second sensor has an extremely poor channel connection. On average, the AoI of the second sensor will be higher than the AoI of the first sensor irrespective of the fraction of times it is scheduled for transmission. As a result, a scheduling algorithm designed to optimize for AoI will attempt to schedule the second sensor at a higher rate in spite of the system conditions limiting its communication capacity. Similarly, the throughput achieved in a collection of sensors that aim at capturing the same physical process albeit with some diversity in measurements, will be low due to the disproportionate amount of resources allocated to the sensors with poor channel conditions. On the other hand, CA-AoI shrinks to zero as the channel conditions deteriorate and the resource allocation will thus be fairer across the sensors. However, using a scheduling policy that blindly maximizes the throughput will be unable to meet the timeliness requirements of the system. As we discuss ahead in Section \ref{sec:discuss}, a scheduling algorithm designed to optimize for CA-AoI will be better suited to balance throughput and freshness compared to AoI.


Optimizing this objective has several challenging aspects. For most sources the channel state information (CSI) is unknown. Hence the scheduler needs to guess when the channel state is ON so that the source can push an update. On the other hand, for some sources CSI is known before every time-slot. But even with CSI, designing a scheduling policy is highly nontrivial because the scheduler is allowed to schedule only a single source in every time slot. Computing the optimal policy using dynamic programming is difficult for such problems as the state-space grows exponentially with the number of sources. Hence low complexity policies are required that can make this constrained scheduling problem computationally efficient. 

\begin{figure}[t]
\centering
    \includegraphics[scale = 0.40]{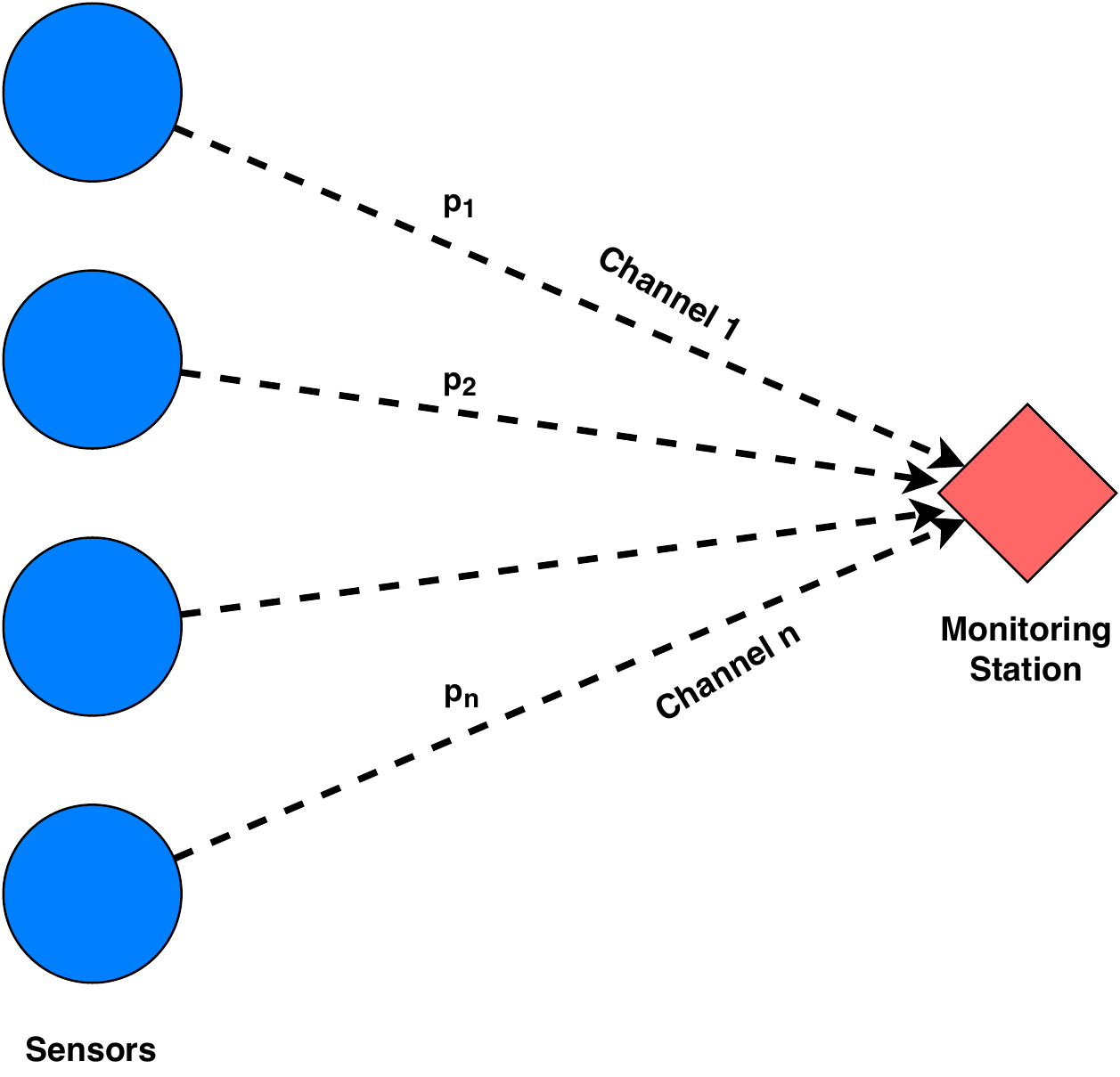}
    \caption{A multi-source system consisting of $n$ sensors communicating with a central monitoring station via unreliable communication channels. $p_{i}$ denotes the channel ON probabilities. }
    \label{fig:PS}
\end{figure}

\subsection{Our Contributions}
 
We consider three cases, namely, systems without CSI, systems with CSI, and systems with partial CSI. 
 \begin{itemize}
     \item[--] We propose a new AoI-like freshness metric called Channel Aware-Age of Information (CA-AoI). Unlike vanilla AoI, CA-AoI is not agnostic of channel reliability. 
     \item[--] We derive universal lower bounds on the average weighted sum of CA-AoI for different CSI models.  
      \item[--] We view the scheduling problem through the lens of a restless multi-armed bandit framework (RMAB) and prove that all three problem instances are indexable. We derive structural properties for the optimal policy. Therefore we calculate closed form expressions for the Whittle index and propose low complexity scheduling algorithms. 
     \item[--] We propose a suite of stationary randomized scheduling algorithms to optimize the objective for all three problem instances. We find closed form expressions for the cost under these policies and thus optimal parameters for the algorithms. 
     \item[--] We compare the performance of the Whittle index and randomized policies against the greedy policy via simulations. We show the efficacy of our scheduling algorithms and conclude the superiority of the Whittle index policy in diverse settings.
     \item[--] We finally demonstrate the advantage of using the sum of CA-AoIs over the sum of AoIs as the freshness metric in handling the diversity-throughput trade off. 
 \end{itemize}

\subsection{Related Work}

As mentioned above, there exists a rich body of work that uses AoI minimization in scheduling applications \cite{7492912}-\cite{8514816}. Several existing works model the scheduling problem as restless bandits and use Whittle index policies to solve them. The work in \cite{9027444} derives the Whittle index for AoI minimization in a system having iid/Markovian channels and multiple information settings (without CSI/ with CSI/ delayed CSI).  In \cite{article}, the authors use Whittle's methodology for the setting where realization of the stochastic arrival process for a given time-slot is known apriori. In \cite{8486307}, the focus is on designing Whittle index policies for AoI minimization when the system has throughput constraints. In \cite{8919842}, the authors consider Whittle index policies to minimize non-decreasing functions of AoI and prove optimality of the index policy for a system with 2 sources and reliable channels. Moreover in \cite{8514816}, a lower bound for the sum of AoIs metric is derived and various low complexity algorithms with theoretical guarantees are designed. Very recently in \cite{NIPS2019_8348}, the authors focus on scheduling a web engine under bandwidth constraints to track content changes in sources with different observabilities  and propose randomized algorithms to optimize a harmonic freshness objective. 

\section{Setting}
 We consider a system of \textit{n} sensors, each communicating with a central monitoring station via a shared unreliable communication channel. Each sensor measures a time varying quantity and can push real time updates to the central monitoring station. 
 
 \subsection{Communication Model}
 Time is divided into slots. In a time slot, a maximum of a single sensor can push updates to the monitoring station. Here the sensors are assumed to be active sources and can thus generate updates in every time slot. 
\subsection{Channel Model}
We model the channels as ON-OFF channels with channel ON probability for channel $i$ denoted as $p_{i}$. Here we assume that channel realizations are independent across sensors, and independent and
identically distributed (i.i.d.) over time-slots. 

\subsection{Channel State Information (CSI) Models}
We consider the following three cases:
\begin{enumerate}
    \item[--] \emph{Systems without CSI}: In this case, the state of the channels for the various sensors (ON/OFF) in a time-slot is not known to the scheduler before making the scheduling decision for that time-slot.  
    \item[--] \emph{Systems with CSI}: In this case, the state of the channels for the various sensors (ON/OFF) in a time-slot is known to the scheduler before making the scheduling decision for that time-slot.
    \item[--] \emph{Systems with partial CSI}: In this case, the state of the channels for some (a fixed subset) of the sensors is known to the scheduler before making the scheduling decision for that time-slot.
\end{enumerate}
In all three cases, channel statistics are known to the scheduler.
\subsection{Our Metric: Channel-Aware Age of Information (CA-AoI)}

Let $X_{i}(t)$ denote the CA-AoI of sensor $i$ at time $t$. We denote the channel state at time $t$ with respect to sensor $i$ as $\Lambda_{i}(t)$, where $\Lambda_{i}(t)$ = 1 means that the channel is ON and $\Lambda_{i}(t)$ = 0 implies that the channel is OFF. We indicate the action taken for sensor $i$ with $a_{i}(t)$, where $a_{i}(t)$ = 1 means that sensor $i$ is scheduled as time $t$ and $a_{i}(t)$ = 0 otherwise. Given the above definitions, CA-AoI for a sensor evolves as follows:
\begin{align*}
X_{i}(t+1) = \begin{cases}
0, & \text{if } a_{i}(t) = 1 \text{ and } \Lambda_{i}(t) = 1,\\
X_{i}(t)+1, & \text{if } a_{i}(t) = 0 \text{ and } \Lambda_{i}(t) = 1, \\
X_{i}(t), & \text{otherwise}.
\end{cases}
\end{align*}
Note that as opposed to the vanilla AoI metric, CA-AoI doesn't penalize the age of a sensor if the channel is OFF at time $t$. In other words the age of a sensor is incremented only when the channel is in the ON-state and the sensor isn't scheduled.

\subsection{Objective}
The objective of the scheduling policy $\pi$ is to minimize the time average of the weighted sum of CA-AoI:
\begin{gather}
    \text{minimize }J^{\pi} = \limsup\limits_{T\rightarrow \infty} \frac{1}{T} \mathbb{E}_{\pi} \left[\sum_{t=1}^{T}\sum_{i=1}^{n} w_{i}X_{i}(t) \right], \label{eq:1}
    \\
    \text{subject to } a_{i}(t) \in \{0,1\} \ \forall i, \sum_{i=1}^{n} a_{i}(t) \leq 1. \nonumber
\end{gather}
In \eqref{eq:1}, $w_{i}$ denotes the weight assigned to sensor $i$ and $\mathbb{E}_{\pi}$ denotes the expected weighted sum of CA-AoIs when a policy $\pi$ is used.

\section{Main results and discussion}

We discuss our key results here. All proofs are provided in the Appendix. We first focus on deriving universal lower bounds for the freshness objective $J^{\pi}$ . 

\subsection{Universal Lower Bounds}

\begin{definition}(Admissible Policies $\Pi$)
In an infinite horizon setting $T\rightarrow \infty$, the class of admissible policies $\Pi$ omits all those policies that stop scheduling a particular sensor after a time slot $T' < \infty$ with a positive probability.
\end{definition}
Our first result characterizes a lower bound on the solution of $J^{\pi}$ for a system without CSI among the class of admissible policies $\Pi$. We denote the set of sensors in this system by $n^{-}$.
\begin{theorem}
\label{theorem:lower_bound_ucsi}
For a system without CSI,  $J^{\pi^{-}} \geq \mathcal{L}^{-}\quad \forall \pi^{-} \in \Pi$  where

\begin{equation}
    \mathcal{L}^{-} = \frac{(\Sigma_{i \in n^{-}} \:\sqrt{w_{i}p_{i}})^{2} - \Sigma_{i \in n^{-}} w_{i}p_{i} }{2}  \nonumber
\end{equation}
\end{theorem}

The proof of Theorem \ref{theorem:lower_bound_ucsi} can be found in Appendix \ref{appe:lower_bound_ucsi}. The proof follows a sample path argument and uses Fatou's lemma and is derived from the technique in \cite{8514816} used to find a lower bound for the AoI objective. 

In the next result we derive the lower bound on  $J^{\pi}$ for a system with CSI having a set of sensors $n^{+}$, among the class of admissible policies $\Pi$. 

\begin{theorem}
\label{theorem:lower_bound_csi}
For a system with CSI, $J^{\pi^{+}} \geq \max\left(0,\mathcal{L}^{+}\right) \quad \forall \pi^{+} \in \Pi$  where

\begin{equation}
    \mathcal{L}^{+} = \frac{(\Sigma_{i \in n^{+}} \:\sqrt{w_{i}}p_{i})^{2} - \Sigma_{i \in n^{+}} w_{i}p_{i} }{2}  \nonumber
\end{equation}
\end{theorem}
The proof of Theorem \ref{theorem:lower_bound_csi} can be found in Appendix \ref{appe:lower_bound_csi}. \footnote{For a system with partial CSI  the lower bound can be found by summing the lower bounds derived on sensors with and without CSI.
Note that it may be possible to further tighten these lower bounds.}

\subsection{Whittle Index Scheduling Policies}
The scheduling problem can be viewed as a Restless Multi Armed Bandit (RMAB) with the sensors as arms whose age evolve even when they are not scheduled for transmission. Solving the RMAB problem is P-SPACE hard. However relaxing the constraints on the number of arms that can be played in every time slot and instead placing a constraint on its expectation, makes this problem tractable. The relaxed problem helps decouple the problem into $n-$subproblems, one for each sensor and the solution of this problem is called the Whittle index policy. Hence each subproblem consists of a single sensor, a single channel and a playing charge c (Lagrange multiplier) for the arm. Let $c(s(t),a(t))$ denote the cost incurred by action $a(t)$ in state $s(t)$.   A policy $\pi = \{a(1),a(2) \cdots\}$ decides whether to idle the arm or play the arm at each time step. A policy is cost-optimal if it minimizes the average cost, defined as follows:
\[\limsup_{T\rightarrow\infty}\frac{1}{T}\sum_{t=1}^{T}\mathbb{E}_{\pi}\left[c(s(t),a(t))\right|s(0)=0].\]
\begin{definition} (Indexability)
Given a playing charge c, let S(c) be the set of states for which the optimal policy is to idle. Then the subproblem is called \textit{indexable} if S(c) monotonically increases from the empty state to the entire state space as c increase from $-\infty \text{ to } \infty$. 
\end{definition}
\begin{definition}(Whittle index)
The Whittle index of a state is the playing charge that makes it equally lucrative to play or idle the arm.
\end{definition}
We now focus on three cases, namely, systems without CSI, systems with CSI, and systems with partial CSI, separately.
\subsubsection{System without CSI} Here the state of the channel is not known a priori and can be determined only after scheduling an information packet via the channel. We first formulate the decoupled subproblem. We drop the sensor index $i$ for further analysis as each sub-problem deals with a single sensor.\\
\textit{Decoupled subproblem}:\\
\textit{States}: The CA-AoI ($X(t)$) of the source defines its state $s(t)$\\
\textit{Actions}: The actions $\in\{0,1\}$ are denoted by $a(t)$, where $a(t)=1$ indicates that sensor is scheduled to update the central station and $a(t)=0$ signifies that the sensor is kept idle.\\
\textit{Transition probabilities}:
The transition probability to state $s'$ after taking action a from state $s=x$ is:
\begin{align*}
&\mathbb{P}(s' = x | a=1,s=x) = 1-p\\
&\mathbb{P}(s' = 0 | a=1,s=x) = p\\
&\mathbb{P}(s' = x | a=0,s=x) = 1-p\\
&\mathbb{P}(s' = x+1 | a=0,s=x) = p
\end{align*}
\textit{Cost}: Let $\Lambda(t) \in \{0,1\}$ denote the channel state. Then the cost incurred by choosing action a at state $s = x$ is:
\[c(s,a) = w[x(1-a\Lambda)+\Lambda(1-a)]+ca.\]
The above cost function can be simply deduced as the sum of the post-action age and the charge incurred when playing an arm.

The Markov Decision Process (MDP) for this subproblem has been shown in Fig. \ref{fig:mdp} (a). Our next result shows that systems without CSI are indexable and characterizes the Whittle index in closed form.

\begin{theorem}
\label{theorem:UCSI_whittle}
 The scheduling problem for a system without CSI is indexable and the Whittle index of a sensor in state x is given by:
 \[\mathcal{W}(x,p) = \frac{w(x+1)(x+2)}{2(2-p)}\]
\end{theorem}

We use this result to design a Whittle index based scheduling policy which is to schedule the sensor with the highest Whittle Index in each time-slot. See Appendix \ref{appe:Th1} for the proof.\footnote{Note that it is not possible to know the channel states of the unscheduled sensors, under the system without CSI. To update the CA-AoI of an unscheduled sensor $i$ we set it's channel state $\Lambda$ to be the outcome of a Bernoulli random variable with mean $p_i$.}

\subsubsection{System with CSI}
In this case the scheduler knows the states of the channel before scheduling. We begin by describing the subproblem for this case:\\
\textit{Decoupled subproblem}:\\
\textit{States}: The state $s(t)$ of the arm is defined by the tuple $(X(t),\Lambda(t))$\\
\textit{Actions}: $a(t)=1$ indicates that sensor is scheduled to update the central station and $a(t)=0$ signifies that the sensor is kept idle.\\
\textit{Transition probabilities}: The transition probability to state $s'$ after taking action a from state $s=x$ is:
\begin{align*}
&\mathbb{P}(s' = (x,0) | a=0,s=(x,0)) = 1-p\\
&\mathbb{P}(s' = (x,1) | a=0,s=(x,0)) = p\\
&\mathbb{P}(s' = (x,0) | a=1,s=(x,0)) = 1-p\\
&\mathbb{P}(s' = (x,1) | a=1,s=(x,0)) = p\\
&\mathbb{P}(s' = (x+1,0) | a=0,s=(x,1)) = 1-p\\
&\mathbb{P}(s' = (x+1,1) | a=0,s=(x,1)) = p\\
&\mathbb{P}(s' = (0,0) | a=1,s=(x,1)) = 1-p\\
&\mathbb{P}(s' = (0,1) | a=1,s=(x,1)) = p.
\end{align*}

\begin{figure}[t]
    \centering
    \subfloat[System without CSI]{{\includegraphics[width=0.5\linewidth]{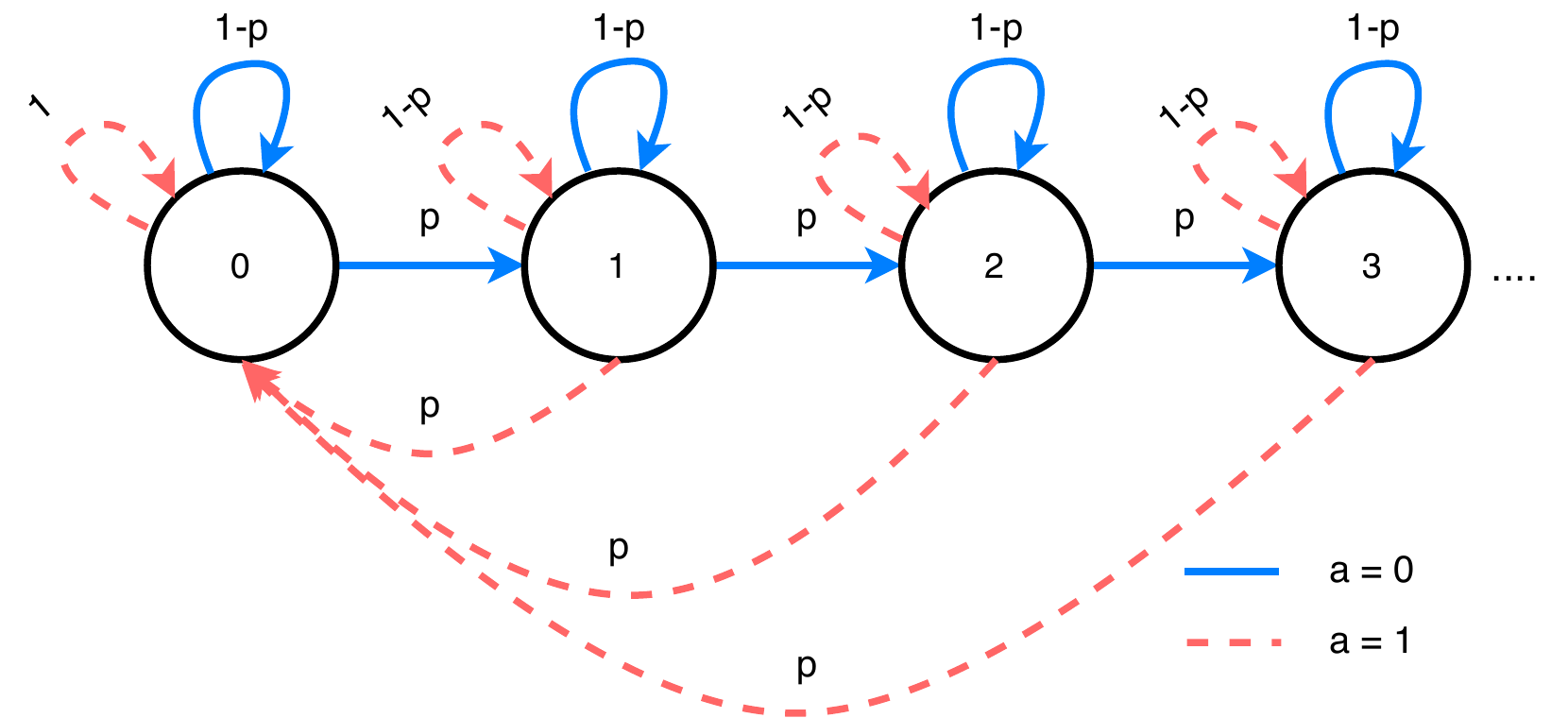} }}%
    \qquad
    \subfloat[System with CSI]{{\includegraphics[width=.4\linewidth]{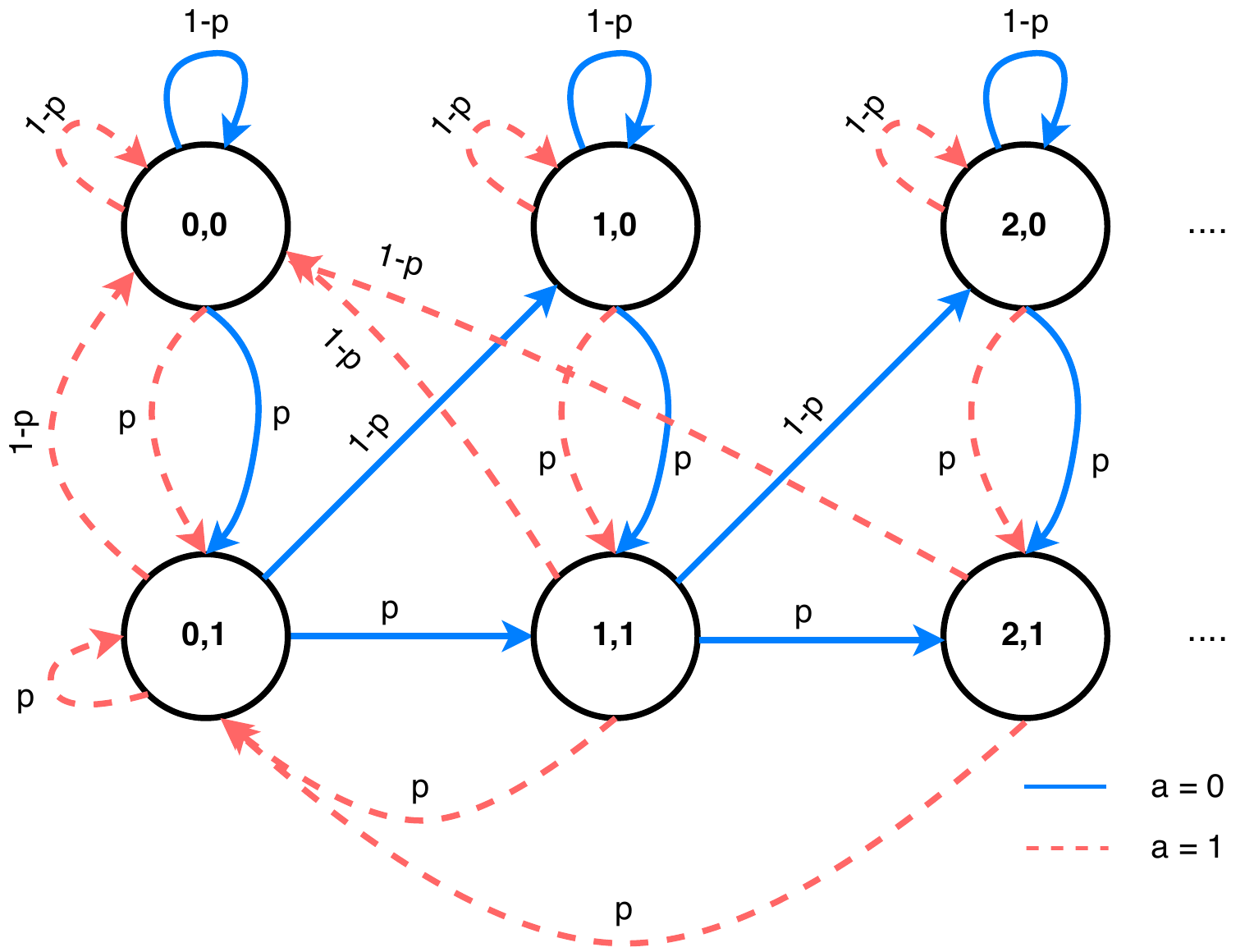} }}%
    \caption{MDP for decoupled subproblems}%
    \label{fig:mdp}
\end{figure}
\textit{Cost}:  The cost incurred by choosing action a at state $s = x$ is:
\[c(s,a) = w[x(1-a\Lambda)+\Lambda(1-a)]+ca.\]
The MDP for this subproblem has been represented in Fig. \ref{fig:mdp} (b). Our next result shows that systems with CSI are indexable and characterizes the Whittle index in closed form.

\begin{theorem}
\label{theorem:CSI_whittle}
 The scheduling problem for a system with CSI is indexable and the Whittle index of a sensor in state (x,$\Lambda$) is given by:
 \begin{align*}
 \mathcal{W}(x,p) = \begin{cases}
 \dfrac{w(x+1)(x+2)}{2} & \text{if } \Lambda = 1,\\
 0 & \text{otherwise.}
 \end{cases}
 \end{align*}
\end{theorem}
We use this result to design a Whittle index based scheduling policy which is to schedule the sensor with the highest Whittle Index in each time-slot. Note that the Whittle index gives us an online algorithm because it is independent of the channel statistics.  See Appendix \ref{appe:Th2} for the proof. 

\subsubsection{System with partial CSI}
We now consider the case for which CSI is available only for a subset of the sensors. We refer to this setting as systems with partial CSI. As in the previous two cases, the scheduling problem reduces to $n$ subproblems and the previous results hold true for this case as well. Let $n^{-}$ and $n^{+}$ be the set of sensors without CSI and with CSI respectively.  
Our next result shows that systems with partial CSI are indexable and characterizes the Whittle index in closed form. This result follows from Theorems \ref{theorem:UCSI_whittle} and \ref{theorem:CSI_whittle}.
\begin{theorem}
\label{theorem:mixed_whittle}
 The scheduling problem for a system having partial CSI is indexable and the Whittle index of a sensor in state x is given by:
 \begin{align*}
  \mathcal{W}(x_{i},p_{i}) = \begin{dcases}
 \frac{w_{i}(x_{i}+1)(x_{i}+2)}{2(2-p_{i})} & \text{if } i \in n^{-},\\ 
 \frac{w_{i}(x_{i}+1)(x_{i}+2)}{2} & \text{if } \Lambda_{i} = 1, i\in n^{+},\\
 0 & \text{otherwise.}
 \end{dcases}
 \end{align*}
\end{theorem}

\subsection{Randomized Scheduling Policies}
We refer to a policy as a randomized scheduling policy if it specifies a probability
distribution on the set of scheduling decisions at every time slot. To characterize the optimal\footnote{Optimal among policies belonging to the given class of randomized scheduling policies.} randomized scheduling policy $\pi^{*}$ for each of the three cases, we first find closed form expressions for the objective under a class of randomized policies and give algorithms to find policy parameters that optimize the objective.

\subsubsection{System without CSI}
 The randomized policy $\pi^{-}$ in this case is to schedule source $i$ according to a Bernoulli process with parameter $\Delta_{i}$. Let $n^{-}$ denote the set of all sensors without CSI and $\vec{\Delta}$ the vector of parameters $\Delta_{i}$. Then the policy is:
\begin{gather*}
    \pi^{-} = \{\text{Sched}(i) \sim \text{Bernoulli}(\Delta_{i}) \quad \forall  i \in n^{-} | 1 \geq \Delta_{i} > 0 \}.
\end{gather*}

To optimize the parameters $\Vec{\Delta}$ of the policy we relax the constraints on $\pi^{-}$ from $\sum_{i=1}^{n} a_{i}(t) = 1$ to $\mathbb{E}\left[\sum_{i=1}^{n} a_{i}(t)\right] = 1.$
Our next result characterizes the performance of a randomized policy as a function of the parameters $\vec{\Delta}$.
\begin{theorem}
\label{theorem:UCSI_random}
Without loss of generality assume that $\vec{p}>0$. Then under the policy $\pi^{-}$, the freshness objective $J^{\pi^{-}}$ as defined in Eq. \ref{eq:1} is 
\begin{equation}
   J^{\pi^{-}} = \sum_{i \in n^{-}}\frac{w_{i}(1-\Delta_{i})}{\Delta_{i}}.
\end{equation}
\end{theorem}

Check Appendix \ref{appe:Th4} for the proof. Next, we use this result to design an optimization problem for finding the optimal parameters $\vec{\Delta}$.

\textit{Problem 1: Finding optimal $\pi^{*} \in \pi^{-}$} -- Finding the best policy reduces to the following convex optimization problem.
\begin{gather*}
    \text{minimize }J^{\pi^{-}} = \sum_{i\in n^{-}}\frac{ w_{i}(1-\Delta_{i})}{\Delta_{i}},
    \\
    \text{subject to } \sum_{i\in n^{-}} \Delta_{i} = 1 , 1\geq \Delta_{i}> 0 \quad \forall i \in n^{-}. \nonumber
\end{gather*}

Note that the first constraint comes from the relaxation of the scheduling problem constraints. Using the Lagrange multiplier method, finding the optimal policy reduces to solving the following system of equations:
\begin{gather}
    \Delta_{i} = \sqrt{\frac{w_{i}}{\lambda}} \quad \forall i \in n^{-},\label{eq_deltas}\\
    \sum_{i\in n^{-}} \Delta_{i} = 1. \label{eq_lamda}
\end{gather}


Note that an exact solution to $\vec{\Delta}$ can be obtained from the above system of equations. The optimization procedure is given in Algorithm \ref{algo_1}.

\IncMargin{1em}
\begin{algorithm}
\SetKwInOut{Input}{input}\SetKwInOut{Output}{output}
\Input{$p_{i}$ - Channel ON probability $\forall i \in n^{-}$ \\ $w_{i}$ - Weight of sensor $i$ $\forall i \in n^{-}$ }
\Output{$\vec{\Delta}^{*}$ }
\BlankLine
$\sqrt{\lambda^{*}}$ $\longleftarrow \: \sum_{i \in n^{-}} \sqrt{w_{i}}$ \;
 \For{i $\in n^{-}$}{
$\Delta_{i}^{*}\longleftarrow \frac{\sqrt{w_{i}}}{\sqrt{\lambda^{*}}}$
}
\Return{$\vec{\Delta}^{*}$} 
\caption{Finding $\pi^{*} \in \pi^{-}$ (Problem 1)}\label{algo_1}
\end{algorithm}\DecMargin{1em}

 See Appendix \ref{appe:alg1} for more details. Hence we schedule sensor $i$ with a probability $\Delta^*_i$.

\subsubsection{System with CSI} For systems with CSI, the scheduling problem is again non-trivial since only one sensor can be scheduled in every time slot. Hence we schedule each sensor for only a fraction of the times for which the channel is in the ON-state. Let $n^{+}$ denote the set of sensors with CSI and $\vec{\alpha}$ denote the vector of parameters $\alpha_{i}$. The policy $\pi^{+}$ is:
\begin{gather*}
    \pi^{+} = \{\text{Schedule sensor }i\text{ with probability }\alpha_{i}\text{ if channel state = ON } \: \forall i \in n^{+} | 1 \geq \alpha_{i} > 0. \}
\end{gather*}
 Like the previous case we find an optimal policy for relaxed constraints. Our next result characterizes the performance of a randomized policy as a function of the parameters $\vec{\alpha}$.
 \begin{theorem}
 \label{theorem:CSI_random}
Without loss of generality assume that $\vec{p}>0$. Then under the policy $\pi^{+}$, the freshness objective $J^{\pi^{+}}$ as defined in Eq. \ref{eq:1} is 
\begin{equation}
   J^{\pi^{+}} = \sum_{i\in n^{+}}\frac{ w_{i}(1-\alpha_{i})}{\alpha_{i}}.   
\end{equation}
 \end{theorem}
 Check Appendix \ref{appe:Th5} for the proof. Note that sensor $i$ has channel ON probability $p_{i}$ and is scheduled with probability $\alpha_{i}$ when channel state is ON. Thus the average rate at which it is scheduled is $p_{i}\alpha_{i}$. This problem can also be setup as an optimization problem.

\textit{Problem 2: Finding optimal $\pi^{*} \in \pi^{+}$} -- Finding the optimal policy can be reduced to the following convex optimization problem.
\begin{gather*}
    \text{minimize }J^{\pi^{+}} = \sum_{i\in n^{+}}\frac{ w_{i}(1-\alpha_{i})}{\alpha_{i}}
    \\
    \text{Subject to } \sum_{i\in n^{+}} p_{i}\alpha_{i} = 1 ,\:1 \geq \alpha_{i}> 0 \quad \forall i \in n^{+}.
\end{gather*} 

These constraints are derived from the relaxation of the constraints of the original scheduling problem. We ignore the $\alpha_{i}\leq 1$ constraint and optimize this using the Lagrange multiplier method. We describe how we account for this omitted constraint ahead. The optimal solution is obtained by solving the following system of equations:
\begin{gather}
    \alpha_{i} = \sqrt{\frac{w_{i}}{p_{i}\lambda}} \quad \forall i \in n^{+}, \label{csi_opt1}\\
    \sum_{i\in n^{+}} p_{i}\alpha_{i} = 1. \label{eq_csi}
\end{gather}

An exact solution of the above set of equations can be found. Algorithm \ref{algo_2} describes the optimization procedure.\\
\IncMargin{1em}
\begin{algorithm}
\SetKwData{R}{R}\SetKwData{ActiveSet}{ActiveSet}\SetKwData{violation}{violation}
\SetKwInOut{Input}{input}\SetKwInOut{Output}{output}
\Input{$p_{i}$ - Channel ON probability $\forall i \in n^{+}$ \\ $w_{i}$ - Weight of sensor $i$ $\forall i \in n^{+}$}
\Output{$\vec{\alpha}^{*}$ }
\BlankLine
\R$\longleftarrow 1$\;
\ActiveSet$\longleftarrow n^{+}$\;
\BlankLine
\While{\ActiveSet $\neq \phi$ }{
\violation $\longleftarrow$ False\;
\BlankLine
$\sqrt{\lambda^{*}}$ $\longleftarrow \sum_{i \in {ActiveSet}}\frac{\sqrt{w_{i}p_{i}}}{R}$ \;

\For{i $\in$ \ActiveSet}{
$\alpha_{i}^{*} = \sqrt{\frac{w_{i}}{p_{i}\lambda^{*}}}$\;
\uIf{$\alpha_{i}^{*} \geq$ 1}{
$\alpha_{i}^{*} = 1$\;
\violation = True\;
\R= \R $-\:p_{i}$\;
\ActiveSet = \ActiveSet $\backslash \{i\}$
}
}
\uIf{\violation == False}{Break}
}

\Return{$\vec{\alpha}^{*}$} 
\caption{Finding $\pi^{*} \in \pi^{+}$ (Problem 2)}\label{algo_2}
\end{algorithm}\DecMargin{1em}
In Algorithm \ref{algo_2}, $\sqrt{\lambda^{*}}$ is obtained by solving the following set of equations:\\
\begin{gather}
    \alpha_{i} = \sqrt{\frac{w_{i}}{p_{i}\lambda^{*}}} \quad \forall i \in ActiveSet\\
    \sum_{i\in ActiveSet} \alpha_{i} p_{i} = R.
\end{gather}
In Algorithm \ref{algo_2}, we initialize the active set with $n^{+}$. If any policy parameter exceeds the inequality constraints i.e. $\alpha_{i}^{*}\geq 1$ then we set $\alpha_{i}^{*}= 1$, remove it from the active set and modify the constraints. This loop is run until the active set is empty or all policy parameters satisfy the inequality constraints. Since the policy found through Algorithm \ref{algo_2} is optimal for the relaxed constraint scheduling problem  ($\mathbb{E}\left[\sum_{i=1}^{n} a_{i}(t)\right] = 1$) we need to make some implementational modification for the original problem ($\sum_{i=1}^{n} a_{i}(t) = 1$). When more than one sensors are to be scheduled according to the policy, we schedule them greedily i.e. schedule the sensor having the highest weighted age. See Appendix \ref{appe:alg2} for more details. 

\subsubsection{System with partial CSI} We combine the policies derived for systems with and without CSI  dealt with in the previous sections to design a scheduling policy for the system having partial CSI. Let $n^{-}$ and $n^{+}$ denote the set of channels without CSI and with CSI respectively. Like the previous cases we relax the constraints and make implementational modifications to account for them. The policy in this case is:
\begin{align*}
\pi= \begin{cases}
\{\text{Schedule(i)} \sim \text{Bernoulli}(\Delta_{i}) \quad \forall  i \in n^{-} | 1 \geq \Delta_{i} > 0 \},\\
\{\text{Schedule(i) with probability }\alpha_{i}\\\text{ if channel state = ON } \: \forall i \in n^{+} | 1 \geq \alpha_{i} > 0 \} .
\end{cases}
\end{align*}

Our next result characterizes the performance of a randomized policy as a function of the parameters $\vec{\Delta}$ and $\vec{\alpha}$.
\begin{theorem}
\label{theorem:mixed_random}
 Without loss of generality assume that $\vec{p}>0$. Then under the policy $\pi$, the freshness objective $J^{\pi}$ as defined in Eq. \ref{eq:1} is 
\begin{equation}
   J^{\pi} = \sum_{i\in n^{-}}\frac{ w_{i}(1-\Delta_{i})}{\Delta_{i}} + \sum_{i\in n^{+}}\frac{ w_{i}(1-\alpha_{i})}{\alpha_{i}}.   
\end{equation}
\end{theorem}
This result is a straightforward extension of Theorems \ref{theorem:UCSI_random} and \ref{theorem:CSI_random}. Hence the scheduling problem can be setup as an optimization problem as follows:

\textit{Problem 3: Finding optimal $\pi^{*} \in \pi$} -- Finding the optimal policy can be reduced to the following convex optimization problem.
\begin{gather*}
    \text{minimize } J^{\pi} = \sum_{i\in n^{-}}\frac{ w_{i}(1-\Delta_{i})}{\Delta_{i}} + \sum_{i\in n^{+}}\frac{ w_{i}(1-\alpha_{i})}{\alpha_{i}},   
    \\
    \text{subject to } \sum_{i\in n^{-}}\Delta_{i} + \sum_{i\in n^{+}} p_{i}\alpha_{i} = 1,\\
    1 \geq \alpha_{i}> 0 \quad \forall i \in n^{+}\:,  1 \geq \Delta_{i}> 0 \quad \forall i \in n^{-}.
\end{gather*} 
We ignore the inequality constraints for $\alpha_{i}$ and use the Lagrange multiplier method to reduce the optimization problem to solving the following set of equations:
\begin{gather}
    \Delta_{i} = \sqrt{\frac{w_{i}}{\lambda}} \quad \forall i \in n^{-},\\
     \alpha_{i} = \sqrt{\frac{w_{i}}{p_{i}\lambda}} \quad \forall i \in n^{+},\\
     \sum_{i\in n^{-}}\Delta_{i} + \sum_{i\in n^{+}} p_{i}\alpha_{i} = 1. \label{eq_3}
\end{gather}

An exact solution to the above system of equations can be found and the optimization procedure is described in Algorithm \ref{algo_3}. 

\IncMargin{1em}
\begin{algorithm}
\SetKwData{R}{R}\SetKwData{ActiveSet}{ActiveSet}\SetKwData{violation}{violation}
\SetKwInOut{Input}{input}\SetKwInOut{Output}{output}
\Input{$p_{i}$ - Channel ON probability $\forall i \in n$ \\ $w_{i}$ - Weight of sensor i $\forall i \in n$}
\Output{$\vec{\Delta}^{*},\vec{\alpha}^{*}$ }
\BlankLine
\R$\longleftarrow 1$\;
\ActiveSet$\longleftarrow n^{+}$\;
\BlankLine
\While{\ActiveSet $\neq \phi$ }{
\violation $\longleftarrow$ False\;
$\sqrt{\lambda^{*}}$ $\longleftarrow \sum_{i \in n^{-}}\frac{\sqrt{w_{i}}}{R}+ \sum_{i \in {ActiveSet}}\frac{\sqrt{w_{i}p_{i}}}{R}$ \;
\BlankLine
\For{i $\in$ \ActiveSet}{
$\alpha_{i}^{*} = \sqrt{\frac{w_{i}}{p_{i}\lambda^{*}}}$\;
\uIf{$\alpha_{i}^{*} \geq$ 1}{
$\alpha_{i}^{*} = 1$\;
\violation = True\;
\R= \R $-\:p_{i}$\;
\ActiveSet = \ActiveSet $\backslash \{i\}$
}
}
\uIf{\violation == False}{Break}
}
\For{i $\in n^{-}$}{
$\Delta_{i}^{*}\longleftarrow \sqrt{\frac{w_{i}}{\lambda^{*}}}$
}
\Return{$\vec{\alpha}^{*},\vec{\Delta}^{*}$} 
\caption{Finding $\pi^{*} \in \pi$ (Problem 3)}\label{algo_3}
\end{algorithm}\DecMargin{1em}

In Algorithm \ref{algo_3} $\sqrt{\lambda^{*}}$ is obtained by solving the following set of equations:
\begin{gather}
    \Delta_{i} = \sqrt{\frac{w_{i}}{\lambda^{*}}} \quad \forall i \in n^{-},\\
     \alpha_{i} = \sqrt{\frac{w_{i}}{p_{i}\lambda^{*}}} \quad \forall i \in \text{ActiveSet},\\
     \sum_{i\in n^{-}}\Delta_{i} + \sum_{i\in \text{ActiveSet}} p_{i}\alpha_{i} = R. \label{eq_4}
\end{gather}
It is straightforward to check that Algorithm \ref{algo_3} utilizes the optimization procedure formulated in Algorithms \ref{algo_1} and \ref{algo_2}. When more than one sensors are to be scheduled according to the policy, we schedule the sensors in a greedy manner.




\section{Simulation Results}
\begin{figure}[t]
\centering
\begin{minipage}{.45\textwidth}
  \centering
    \includegraphics[width=1\linewidth]{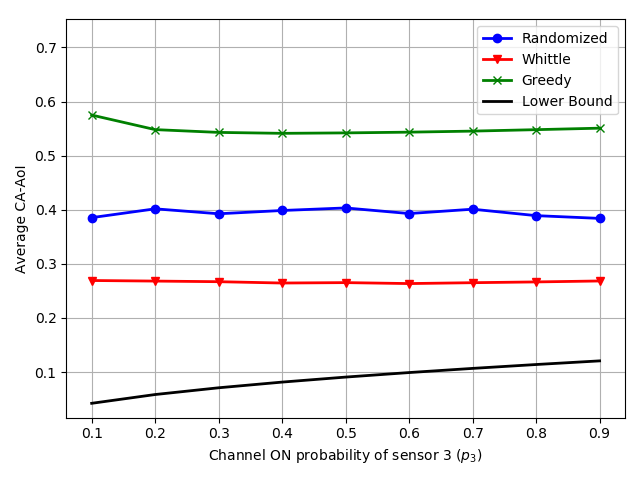}
    \caption{Average sum of weighted age on changing channel ON probability of sensor 3 ($p_{3}$) in a system without CSI }
    \label{fig:UCSI3}
\end{minipage}%
\qquad
\begin{minipage}{.45\textwidth}
 \centering
    \includegraphics[width=1\linewidth]{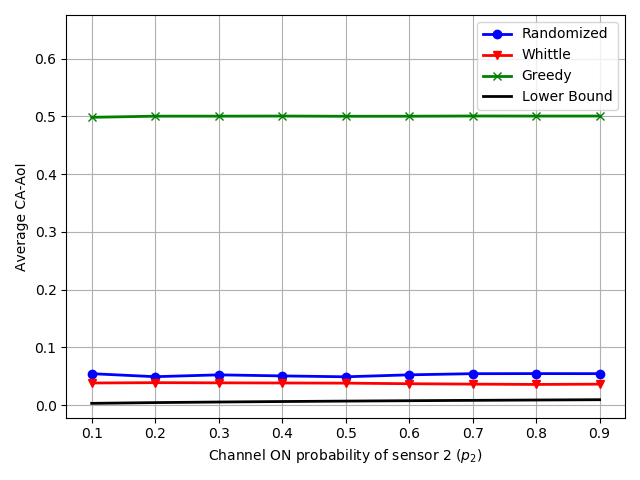}
    \caption{Average sum of weighted age on changing channel ON probability of sensor 2 ($p
    _{2}$) in a system without CSI }
    \label{fig:UCSI2}
\end{minipage}
\end{figure}

\begin{figure}[h]
\centering
\begin{minipage}{.45\textwidth}
  \centering
    \includegraphics[width = 1\linewidth]{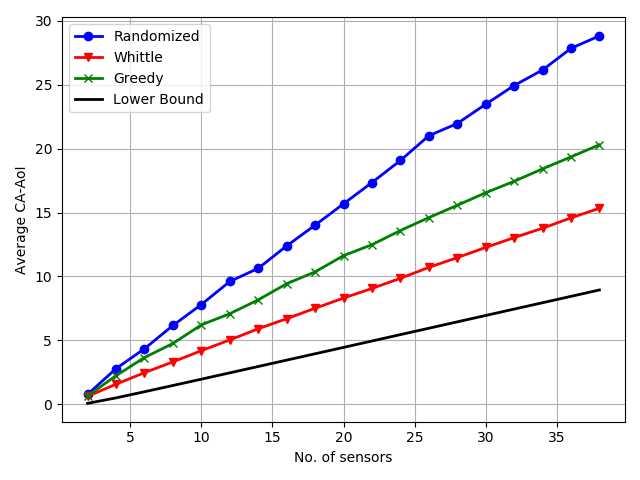}
    \caption{Average sum of weighted age on varying number of sensors in a system without CSI}
    \label{fig:UCSIn}
\end{minipage}%
\qquad
\begin{minipage}{.45\textwidth}
 
\centering
    \includegraphics[width = 1\linewidth]{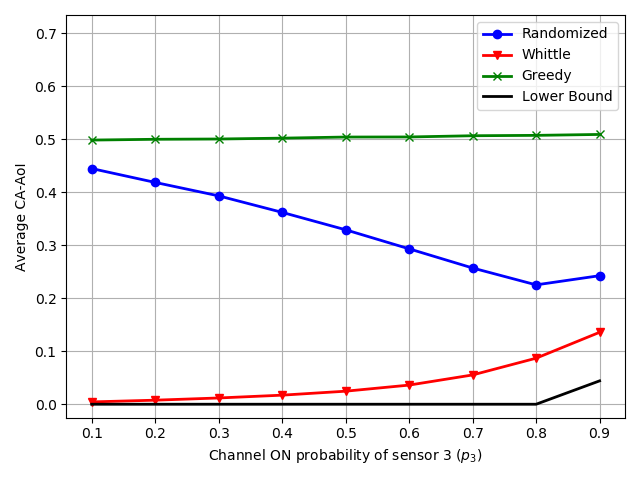}
    \caption{Average sum of weighted age on changing channel ON probability of sensor 3 ($p_{3}$) in a system with CSI }
    \label{fig:CSI3}
\end{minipage}
\end{figure}

\begin{figure}[t]
\centering
\begin{minipage}{.45\textwidth}
  \includegraphics[width = 1\linewidth]{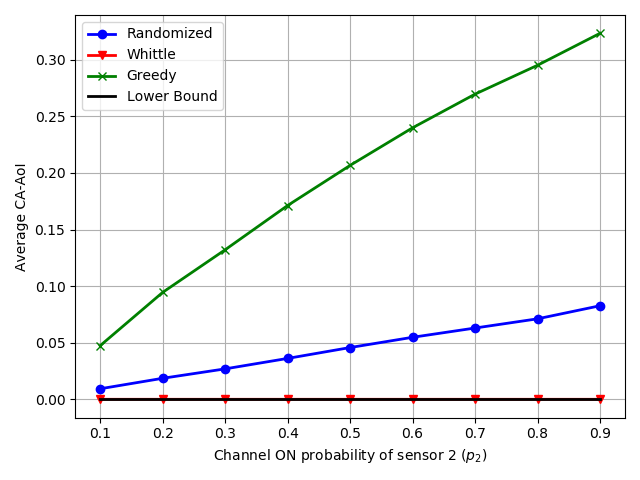}
    \caption{Average sum of weighted age on changing channel ON probability of sensor 2 ($P_{2}$) in a system with CSI }
    \label{fig:CSI2}
\end{minipage}%
\qquad
\begin{minipage}{.45\textwidth}
 
\centering
    \includegraphics[width = 1\linewidth]{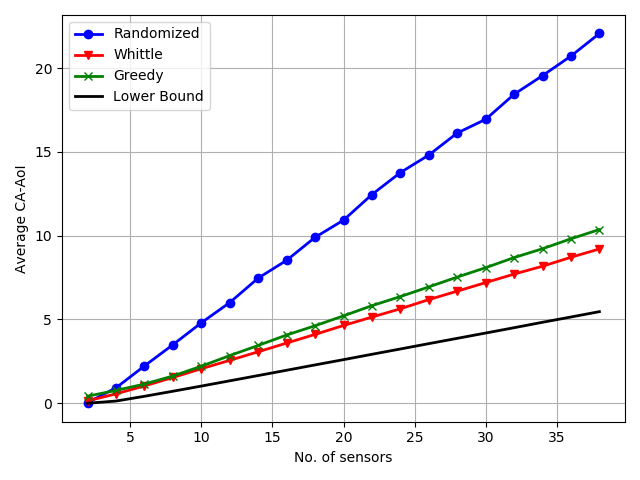}
    \caption{Average sum of weighted age on varying number of sensors in a system with CSI}
    \label{fig:CSIn}
\end{minipage}
\end{figure}

\begin{figure}[t]
\centering
\begin{minipage}{.45\textwidth}
  \includegraphics[width=1\linewidth]{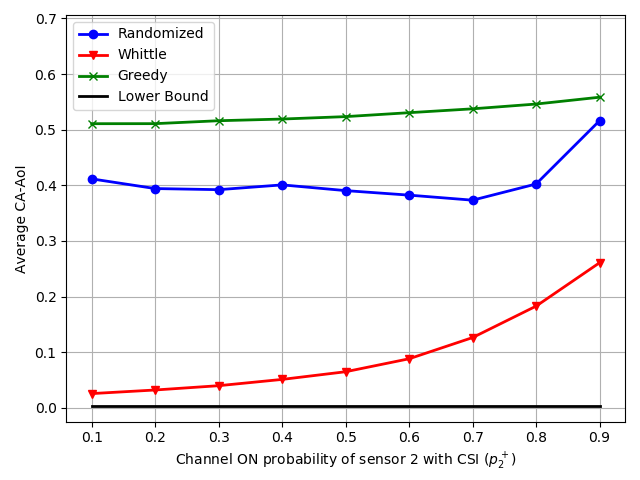}
    \caption{Average sum of weighted age on varying channel ON probability ($p_{2}^{+}$) of sensor 2 with CSI in a partial CSI system}
    \label{fig:mixed_o}
\end{minipage}%
\qquad
\begin{minipage}{.45\textwidth}
 
\includegraphics[width=1\linewidth]{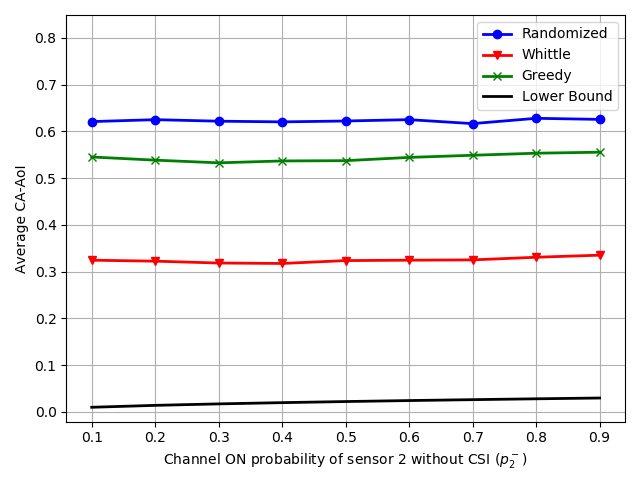}
    \caption{Average sum of weighted age on varying channel ON probability ($p_{2}^{-}$) of sensor 2 without CSI in a partial CSI system}
    \label{fig:mixed_uo}
\end{minipage}
\end{figure}

In this section we compare the performance of our policies against the greedy policy via simulations.
While evaluating the objective we normalize the weights so that they sum to unity. The greedy policy for the system without CSI is to schedule the sensor having the highest weighted age times the channel ON probability in that time slot, i.e., $\argmax_i\left(w_{i}X_{i}(t)p_{i}\right)$. In Fig. \ref{fig:UCSI3} we simulate a system without CSI consisting of three sensors. We set $p_{1}=0.1,\: p_{2} = 0.9$ and vary the channel ON probability of sensor 3 ($p_{3}$) from 0 to 1. The corresponding weights are set as $w_{1}=1,\:w_{2}=1,\:w_{3}=100$. We give sensor 3 a considerably larger weight to clearly observe the effect of changing it's channel ON probability. The Whittle's index policy outperforms the other two policies in this system. In Fig. \ref{fig:UCSI2} we test our policies on a system without CSI consisting of two sensors where an important sensor has a poor connection. We fix $p_{1}=0.1$ and vary $p_{2}$ between 0 to 1. The weights are set as $w_{1} =1000,\: w_{2} = 1$. The Whittle index policy and randomized policy outperform the greedy policy by a considerable margin in this system. In Fig. \ref{fig:UCSIn} we investigate the effect of the size of the system on the performance of the policies. We vary the number of sensors from 10 to 40. In each system the channel statistics are chosen uniformly from 0 to 1. Similarly the weights are chosen uniformly at random from 1 to 100. The Whittle index policy outperforms both the policies.

We run similar experiments for a that system with CSI. The greedy policy for a system having CSI is to schedule the sensor having the highest weighted age i.e. $\argmax_i\left(w_{i}X_{i}(t)\right)$. Fig. \ref{fig:CSI3} shows the simulation of a system with CSI having 3 sensors. The parameters of the system remain the same. Once again the Whittle index policy performs the best and the randomized policy attains similar performance at higher channel ON probabilities. In Fig. \ref{fig:CSI2} we simulate our policies in a system with CSI having important sensors transmitting through poor channels. The parameters remain the same as before. The Whittle index policy and the randomized policy perform substantially better than the greedy policy. We simulate the effect of varying system size on the performance of the policies in Fig. \ref{fig:CSIn}. The greedy policy and the Whittle index policy yield good performance.

In  Fig. \ref{fig:mixed_o} we simulate a system having partial CSI. The system has 2 sensors with CSI and 2 sensors without CSI. The channel statistics of the sensors without CSI are $p_{1}^{-}=0.1,\;p_{2}^{-}=0.9$. The channel ON probability of one of the sensors with CSI is set to $p_{1}^{+} = 0.1$ and the probability of the other sensor is varied from 0 to 1. The weights are setup as $w_{1}^{-}=w_{2}^{-} = w_{1}^{+}=1, \;w_{2}^{+}=100$ . The Whittle index policy performs considerably better than the other policies. In Fig. \ref{fig:mixed_uo} we simulate a similar system albeit wth parameters between sensors with and without CSI exchanged. The channel ON probabilities are $p_{1}^{-}=p_{1}^{+}=0.1,\;p_{2}^{+}=0.9$ and $p_{2}^{-}$ is varied from 0 to 1. The weights are $w_{1}^{+}=w_{2}^{+} = w_{1}^{+}=1, \;w_{2}^{-}=100$ Once again the Whittle index policy outperforms the others. 

From these results, and others which we omit due to lack of space, we conclude that the randomized scheduling policies work well when sensors have poor channel communication. However the Whittle index policy is effective across all types of system settings and achieves minimum cost. It outperforms the randomized policy by a substantial margin in larger systems and systems that have better connections.

\section{Discussion: AoI vs CA-AoI}
\label{sec:discuss}

\label{sec:discuss}
In this section, we discuss why the CA-AoI metric is more suitable than the original AoI metric for some multi-sensor systems. 

As discussed before, the AoI of a sensor increases due to two reasons. The first is when an attempted update by a sensor fails due to poor channel conditions and the second is when the shared channel is used by some other sensor. The CA-AoI distinguishes between these two events by imposing a penalty by increasing the “age” of the information of a sensor at the monitoring station only in the latter case.

Many systems deploy multiple sensors to measure the same quantity, for example  temperature, close to each other to increase robustness to sensor failures as well as diversity. Consider a system with two sensors deployed in the same area. However, due to their placement, one sensor has significantly better channel conditions than the other. This can happen for instance if there is an obstruction between only one of the two sensors and the destination. If the sum AoI metric is used, under any ``good" policy, a larger fraction of resources will be allotted to the sensor with poor channel conditions than if the sum CA-AoI metric is used. As a result, the throughput, i.e., the total number of measurements received by the destination over time will be lower in the first case. Since the quantity being measured by both sensors is roughly the same, CA-AoI might be the more appropriate metric for this case. Sum of CA-AoIs strikes a balance between rewarding the diversity obtained from getting measurements of the same physical process from two independent sources and ensuring that the resource allocation is not heavily tilted towards the sensor with poor channel conditions. On the other hand, using a greedy policy that blindly maximizes the throughput will be unable to meet the freshness requirements of the application. Whittle index policies for either choice of the metric have been demonstrated to work well in the unknown CSI setting. In Fig. \ref{fig:throughput} we show the throughput achieved by Whittle index policies that optimize for AoI and CA-AoI  in an equally weighted two sensor system, where the channel ON probability of the sensor with poor channel conditions is fixed to 0.1 and that of the other sensor is varied. We also simulate the ``greedy'' policy where the sensor with better channel conditions is scheduled in every time slot. We simulate the same setting in Fig. \ref{fig:resources} and measure the fraction of resources allotted to the poor sensor. It can be clearly observed that the throughput obtained by using the CA-AoI metric is higher and the performance gap widens as the channel ON probability of the sensor with better channel conditions increases. In Fig. \ref{fig:age}, we compare the freshness performance obtained from the AoI and the CA-AoI policies via the vanilla AoI metric. As expected, the greedy policy doesn't preserve any timeliness and the Whittle AoI policy outperforms the CA-AoI minimizing policy. While the AoI metric forces a higher scheduling rate of the sensor with poor channel conditions as the probability increases to meet the vanilla freshness requirement, the CA-AoI metric causes much fewer resources to be allocated to the poor sensor. However like the greedy policy, the Whittle CA-AoI policy doesn't completely discount the sensor with poor channel conditions and yields a more fairer allocation of resources and thus freshness of information at the cost of maximum throughput.

\begin{figure}[h]
\centering
\begin{minipage}[b]{.45\textwidth}
  \centering
    \includegraphics[width=1\linewidth]{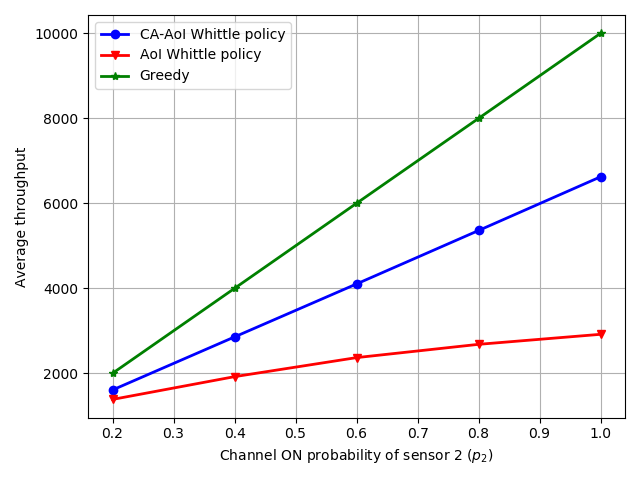}
\end{minipage}%
\qquad
\begin{minipage}[b]{.45\textwidth}
\centering
    \includegraphics[width=1\linewidth]{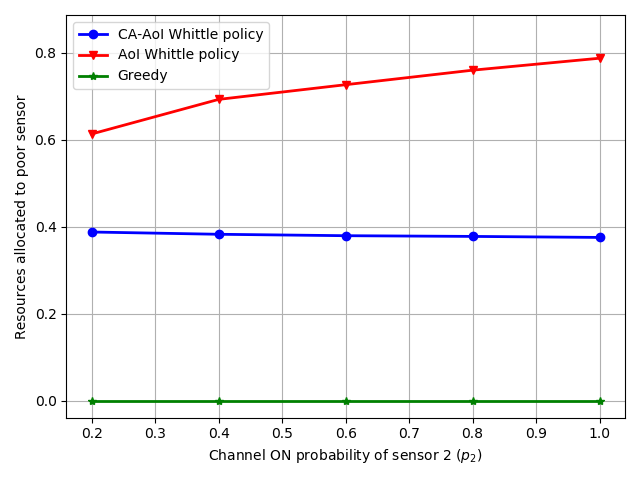}
\end{minipage}
\par
\begin{minipage}[t]{.45\textwidth}
\centering
    \caption{Throughput achieved by minimizing CA-AoI vs AoI}
    \label{fig:throughput}
\end{minipage}
\qquad
\begin{minipage}[t]{.45\textwidth}
\centering
   \caption{Fraction of resources allocated to poor sensor while minimizing CA-AoI vs AoI}
    \label{fig:resources}
\end{minipage}

\end{figure}

\begin{figure}[h]
\centering
\begin{minipage}[b]{.45\textwidth}
  \centering
    \includegraphics[width=1\linewidth]{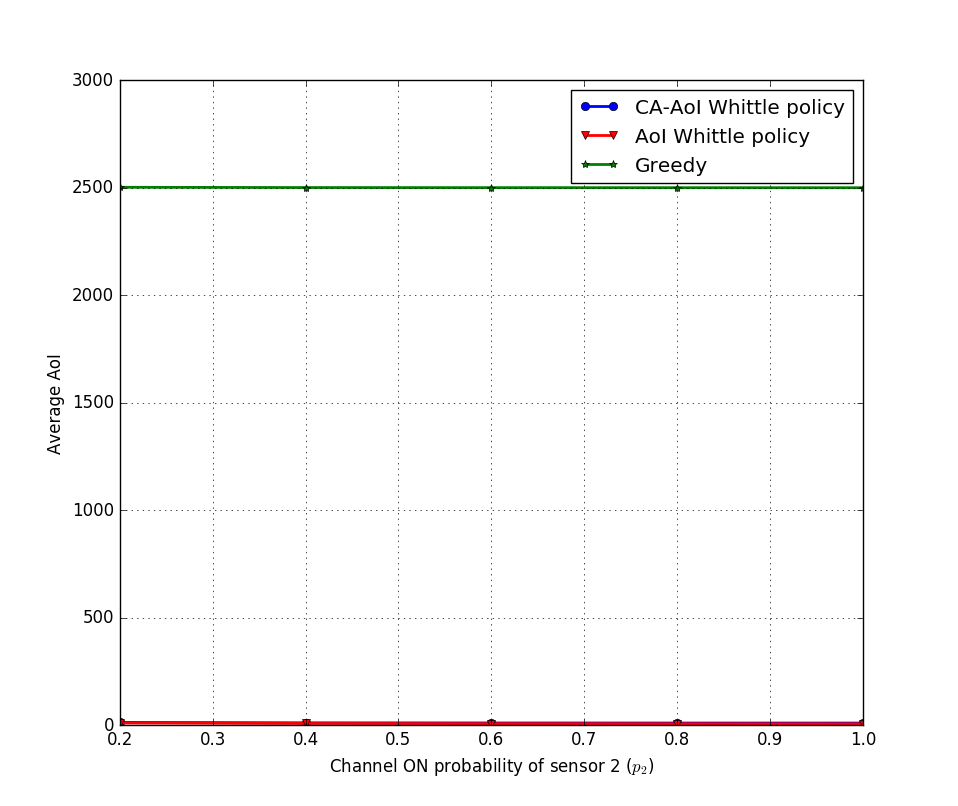}
\end{minipage}%
\qquad
\begin{minipage}[b]{.45\textwidth}
\centering
    \includegraphics[width=1.1\linewidth]{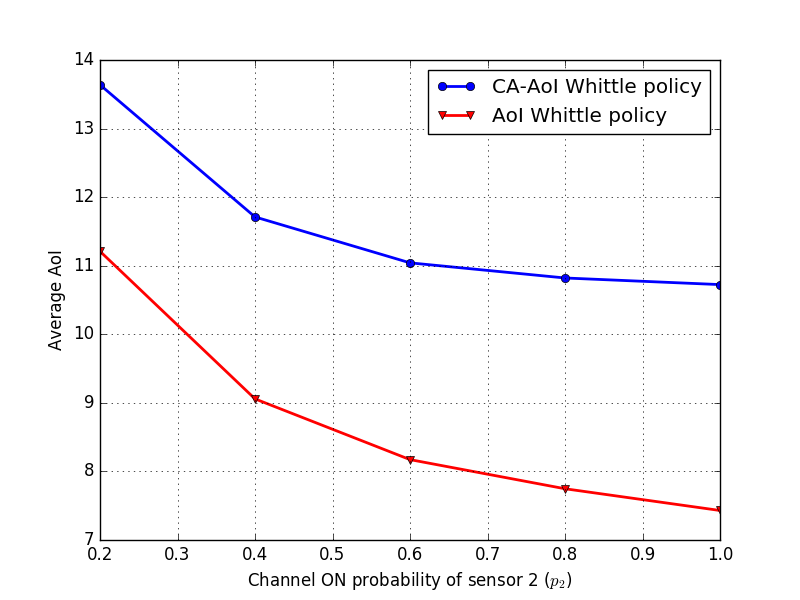}
\end{minipage}
\par
\begin{minipage}[t]{1\textwidth}
\centering
    \caption{Average sum of weighted AoIs while  minimizing CA-AoI vs AoI. Left plot shows all three policies and the right plot compares AoI and CA-AoI Whittle policies}
    \label{fig:age}
\end{minipage}
\end{figure}
\section{Conclusions}
We propose a variant of the AoI metric called Channel-Aware Age of Information (CA-AoI). The CA-AoI of a source at the intended destination is defined as the number of time-slots elapsed since the recent most update from the source was received in which the channel conditions were good enough for the source to send an update. The CA-AoI of a sensor is therefore, a measure of the number of missed opportunities to send an update since its recent most successful update. 
We focus on a multi-source system that updates a monitoring station through a shared unreliable channel under different CSI models and characterize lower bounds on the metric. We model our scheduling problem as a Restless Multi Armed Bandit (RMAB) problem and prove indexability under all the CSI models.  We show the threshold type structure of the optimal policy and derive the Whittle index for the problems. Besides this we also design stationary randomized scheduling policies and show that finding the optimal parameters reduces to simple convex optimization procedures. Via extensive simulations we show that the proposed policies outperform the greedy policies in several settings with the Whittle index policy being the most effective across all system settings. Finally, we discuss the effectiveness of using the sum of CA-AoI metric in a fair resource allocation across sensors and thus in achieving a higher throughput.


\bibliographystyle{./bibliography/IEEEtran}
\bibliography{./bibliography/main}

\appendix
\subsection{Proof of Theorem \ref{theorem:lower_bound_ucsi}}
\label{appe:lower_bound_ucsi}
\begin{figure}[h]

\centering
    \includegraphics[scale = 0.53]{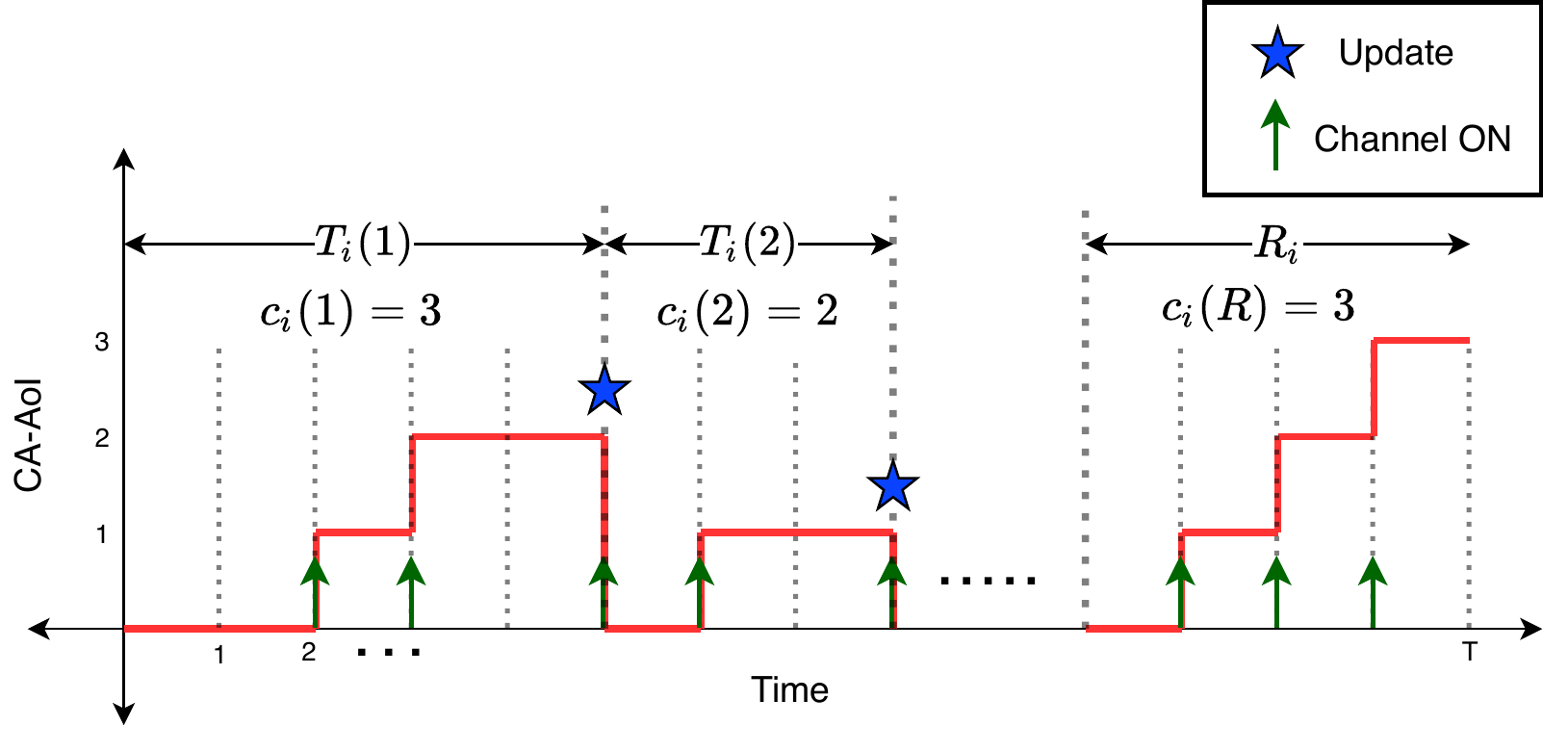}
    \caption{Sample age evolution of a sensor $i$}
    \label{fig:caaoi_evol}
\end{figure}
Consider a sample path $\omega$ that defines the evolution of the ages of $N$ sensors under a policy $\pi$ over a finite time horizon of $T$ slots. Let $D_{i}(T)$ denote the total successful transmissions by sensor $i$ in $T$ slots. Let $T_{i}(m)$ denote the total time slots between the $(m-1)^{th}$ and $m^{th}$ update by sensor $i$ including the time slot of the $m^{th}$ delivery. Let $R_{i}$ be the time horizon after the last successful update. Let $c_{i}(m)$ indicate the total times the channel is in the ON state between the $(m-1)^{th}$ and $m^{th}$ update including the time slot of the $m^{th}$ delivery. Let $c_{i}(R)$ be the total times the channel is on in $R_{i}$. Let $s_{i}$ be the total times the channel is ON throughout the time horizon with respect to sensor $i$. Hence,
\begin{equation}
     s_{i} = \sum_{m=1}^{D_{i}(T)}c_{i}(m)  + c_{i}(R). \label{eq:channel_ON}
\end{equation}

\begin{equation}
      T = \sum_{m=1}^{D_{i}(T)} T_{i}(m) + R_{i}.
      \label{eq:time_T}
\end{equation}
 
Let $X_{i}(t)$ denote the CA-AoI of sensor $i$ at time $t$. The time averaged cost $J^{\pi}_{T}(\geq 0 )$ incurred by the system is:

\begin{equation}
    J^{\pi}_{T} = \frac{1}{T} \sum_{t=1}^{T}\sum_{i=1}^{N} w_{i}X_{i}(t).
\end{equation}

The evolution of age is exactly known from the sample path. In $T_{i}(m)$ slots the age of sensor $i$ evolves from 0 to $c_{i}(m)-1$ before being reset to zero.  The age increments by one in every time slot the channel is ON in $T_{i}(m)$ except for the last time slot where the sensor successfully pushes an update. Hence  $J^{\pi}_{T}$ is lower bounded as follows:

\begin{align} 
   J^{\pi}_{T} &\geq  \frac{1}{T} \sum_{i=1}^{N}w_{i}\{ \sum_{m=1}^{D_{i}(T)}\sum_{n=1}^{c_{i}(m)-1} n \: + \sum_{n=1}^{c_{i}(R)-1} n \},\nonumber\\
   &=  \frac{1}{T} \sum_{i=1}^{N}w_{i}\{\sum_{m=1}^{D_{i}(T)} \frac{c_{i}(m)(c_{i}(m)-1)}{2}+\frac{c_{i}(R)(c_{i}(R)-1)}{2}\}\nonumber,\\
   &\labelrel={seq:channel_ON}  \frac{1}{T} \sum_{i=1}^{N}w_{i}\{\sum_{m=1}^{D_{i}(T)} \frac{c_{i}(m)^2}{2}+ \frac{c_{i}(R)^2}{2} - \frac{s_{i}}{2}\},\nonumber\\
  &= \sum_{i=1}^{N}w_{i} \frac{D_{i}(T)}{T}\{\sum_{m=1}^{D_{i}(T)} \frac{c_{i}(m)^2}{2D_{i}(T)}+ \frac{c_{i}(R)^2}{2D_{i}(T)}\} - \sum_{i=1}^{N} w_{i}\frac{s_{i}}{2T}. \label{eq:lower_bound_exp}
\end{align}
\\

In (\ref{seq:channel_ON}) we use (\ref{eq:channel_ON}). We define the sample mean by the operator $\hat{\mu}$:
\begin{gather}
    \hat{\mu}\left[c_{i}(m)\right] = \frac{1}{D_{i}(T)}\sum_{m=1}^{D_{i}(T)} c_{i}(m). \\
    \hat{\mu}\left[c_{i}(m)^2\right] = \frac{1}{D_{i}(T)}\sum_{m=1}^{D_{i}(T)} c_{i}(m)^2.
\end{gather}

Now using (\ref{eq:time_T}):\\
\begin{gather}
    \frac{T}{D_{i}(T)} = \frac{\sum_{m=1}^{D_{i}(T)} T_{i}(m) + R_{i}}{D_{i}(T)} = \hat{\mu}[T_{i}(m)] + \frac{R_{i}}{D_{i}(T)}. \label{eq:T_dT_ratio}
\end{gather}
Hence on combining (\ref{eq:lower_bound_exp}) and (\ref{eq:T_dT_ratio}) we get:

\begin{multline}
    J^{\pi}_{T} \geq \sum_{i=1}^{N} w_{i} ( \hat{\mu}[T_{i}(m)] + \frac{R_{i}}{D_{i}(T)})^{-1}\{\frac{\hat{\mu}[c_{i}(m)^2]}{2}+\frac{c_{i}(R)^2}{2D_{i}(T)}\} - \sum_{i=1}^{N} w_{i}\frac{s_{i}}{2T}. 
\end{multline}

We simplify the above expression by considering an infinite time horizon $T \rightarrow \infty$. Like \cite{8514816} we define a policy $\pi$ to starve a sensor if it stops scheduling the sensor after a time slot $T^{'}<\infty$ with a positive probability. We exclude policies that starve sensors from the admissible class $\Pi$ without loss of optimality. Since all policies in $\Pi$ push updates repeatedly for all sensors with $p_{i}>0$, it follows that  $R_i, c_{i}(R)$ are finite with a probability one. Hence $D_{i}(T)\rightarrow \infty$, $\frac{R_{i}}{D_{i}(T)}\rightarrow 0$ and $\frac{c_{i}(R)^2}{D_{i}(T)}$ as $T \rightarrow \infty$. Moreover,
\[\lim_{T\rightarrow \infty}\frac{s_{i}}{T} \rightarrow p_{i} \quad w.p.\: 1\]
by the strong law of large numbers. Hence,

\begin{equation} \label{eq:sample_lb}
    \liminf_{T\rightarrow \infty}   J^{\pi}_{T} \geq \sum_{i=1}^{N} w_{i}\frac{\hat{\mu}[c_{i}(m)^2]}{2 \hat{\mu}[T_{i}(m)]}- \sum_{i=1}^{N} w_{i}\frac{p_{i}}{2} \quad w.p.\:1.
\end{equation}

Now,
\begin{equation}
\label{eq:mu_c}
   \lim_{T\rightarrow \infty} \hat{\mu}\left[c_{i}(m)\right] =    \lim_{T\rightarrow \infty} \frac{s_{i}-R_{i}}{D_{i}(T)} = \frac{s_{i}}{D_{i}(T)}.
\end{equation}
We define sample variance operator as $var$:\\
\begin{equation}
var[c_{i}(m)] = \frac{1}{D_{i}(T)}\sum_{m=1}^{D_{i}(T)}(c_{i}(m)-\hat{\mu}[c_{i}(m)])^2.
\end{equation}
$var[c_{i}(m)]$ is a non-negative quantity and can also be written as $\hat{\mu}[c_{i}(m)^2] - \hat{\mu}[c_{i}(m)]^2 $. We denote the total attempts at updating by sensor $i$ by $A_{i}(T)$. Since a maximum of single sensor is allowed to be scheduled in a timeslot:
\begin{equation}\label{eq:attempts}
   \sum_{i=1}^{N} A_{i}(T) \leq T.
\end{equation}
Since the probability of a successful transmission is $p_{i}$ by strong law of large numbers:
\begin{equation}\label{eq:ad}
    \lim_{T\rightarrow\infty} \frac{D_{i}(T)}{A_{i}(T)}\rightarrow p_{i} \quad w.p. \: 1 .
\end{equation}

We further manipulate (\ref{eq:sample_lb}) as follows:\\
\begin{align}
     \liminf_{T\rightarrow \infty}   J^{\pi}_{T} &\labelrel\geq{seq:var} \sum_{i=1}^{N} w_{i}\{\frac{var[c_{i}(m)]}{2 \hat{\mu}[T_{i}(m)]} + \frac{\hat{\mu}[c_{i}(m)]^2}{2 \hat{\mu}[T_{i}(m)]} \} - \sum_{i=1}^{N} w_{i}\frac{p_{i}}{2}\nonumber,\\
     &\labelrel\geq{seq:var_greater} \sum_{i=1}^{N} w_{i}\frac{\hat{\mu}[c_{i}(m)]^2}{2 \hat{\mu}[T_{i}(m)]}- \sum_{i=1}^{N} w_{i}\frac{p_{i}}{2}\nonumber,\\
     &\labelrel={seq:Tinfty} \sum_{i=1}^{N} w_{i}\frac{T}{2D_{i}(T)}\frac{s_{i}^2}{T^2}- \sum_{i=1}^{N} w_{i}\frac{p_{i}}{2}\nonumber,\\
    &\labelrel\geq{seq:attempt} \sum_{i=1}^{N} A_{i}(T) \sum_{i=1}^{N} \frac{w_{i}p_{i}^2}{2D_{i}(T)} - \sum_{i=1}^{N} w_{i}\frac{p_{i}}{2}\nonumber,\\
    &\labelrel\geq{seq:cauchy}\frac{1}{2}\left(\sum_{i}^{N} \sqrt{\frac{w_{i}p_{i}^2A_{i}(T)}{D_{i}(T)}}\right)^2 - \sum_{i}^{N} w_{i}\frac{p_{i}}{2}\nonumber,\\
    &\labelrel={seq:ad} \frac{1}{2} \left(\sum_{i=1}^N \sqrt{w_{i}p_{i}}\right)^2 - \sum_{i}^{N} w_{i}\frac{p_{i}}{2}. \label{eq:lb}
\end{align}
In (\ref{seq:var}) we use the definition of $var[\:]$. In (\ref{seq:var_greater}) we use the fact that $var[\:]$ is positive. In (\ref{seq:Tinfty}) we use (\ref{eq:T_dT_ratio}) as $T\rightarrow\infty$ and (\ref{eq:mu_c}). In (\ref{seq:attempt}) we use (\ref{eq:attempts}). (\ref{seq:cauchy}) uses Cauchy-Schwarz inequality and (\ref{seq:ad}) uses (\ref{eq:ad}) to give $\mathcal{L}^-$ defined in Theorem  \ref{theorem:lower_bound_ucsi}. We now use Fatou's lemma on $J^\pi_T$ to obtain $\limsup_{T \rightarrow \infty}\mathbb{E}[J^\pi_T] \geq \liminf_{T \rightarrow \infty}\mathbb{E}[J^\pi_T] \geq \mathbb{E}[\liminf_{T \rightarrow \infty}J^\pi_T] \geq \mathcal{L}^-.$

\subsection{Proof of Theorem \ref{theorem:lower_bound_csi}}
\label{appe:lower_bound_csi}
The proof technique remains the same except that:
\begin{equation}
     \frac{D_{i}(T)}{A_{i}(T)} = 1. 
\end{equation}
This follows from the fact that the state of the channel is known prior to scheduling. Continuing from step (\ref{seq:cauchy}) of (\ref{eq:lb}):
\begin{align*}
    J^\pi_T  &\geq\frac{1}{2}\left(\sum_{i}^{N} \sqrt{\frac{w_{i}p_{i}^2A_{i}(T)}{D_{i}(T)}}\right)^2 - \sum_{i}^{N} w_{i}\frac{p_{i}}{2},\\
    &= \frac{1}{2}\left(\sum_{i}^{N} \sqrt{w_i}p_i\right)^2 - \sum_{i}^{N} w_{i}\frac{p_{i}}{2}.
\end{align*}
Note that $\mathcal{L}^+$ maybe negative for certain values of $w$ and $p$. However $J^\pi_T$ has a trivial lower bound of 0 and hence we lower bound the freshness objective by $\max(\mathcal{L^+},0)$

\subsection{Proof of Theorem \ref{theorem:UCSI_whittle}}
\label{appe:Th1}
Let $\mathcal{S}$ denote the set of states and $\mathcal{A}$ denote the set of actions. A deterministic policy $\pi(s)$ is a mapping $\mathcal{S}\rightarrow \mathcal{A}$.
\begin{definition}
(Discounted Cost Function) We
define the discounted cost function $V^{\pi}(s)$ as the total expected discounted cost incurred by following the policy from state s.

\begin{align}
V^{\pi}(s;\gamma) &= \limsup_{T\rightarrow\infty}\mathbb{E}_{\pi}\left[\sum_{t=1}^{T}\gamma^{t}C(s(t),\pi(s(t)))|s(0) = s\right].
\end{align}

\end{definition}
A policy $\pi$ is $\gamma-$optimal if it minimizes $V^{\pi}(s;\gamma)$. Let
\begin{gather}
    V^{\gamma}(s) = \min_{\pi} V^{\pi}(s;\gamma).
\end{gather}
The optimal discounted cost function satisfies the following optimality equation:
\begin{gather}
    V^{\gamma}(s) = \min_{a \in A} c(s,a) + \gamma \mathbb{E}[V^{\gamma}(s')].\label{opt}
\end{gather}
\begin{lemma}
\label{finite}
We define an iteration on  $V^{\gamma}_{n}(s)$ with  $V^{\gamma}_{0}(s) = 0 \; \forall s$
\begin{gather}
V^{\gamma}_{n+1}(s) = \min_{a\in\{0,1\}} c(s,a) + \gamma \mathbb{E}[V^{\gamma}_{n}(s')]. \label{value_iterm}
\end{gather}
Then,  $V^{\gamma}_{n}(s)\rightarrow V^{\gamma}(s)$ as $n\rightarrow\infty$
\end{lemma}
\begin{proof}
It is sufficient to show that $V^{\gamma}(s) < \infty \quad \forall s, \forall \gamma$ \cite{10.2307/171262}\\
Consider the deterministic policy $\pi$ of scheduling the sensor in every state. Then under $\pi$, $V^{\pi}(s;\gamma)$ for $s=x$ is:
\begin{gather}
V^{\pi}(s;\gamma)= (1-p)(c+wx+\gamma V^{\pi}(s;\gamma))+p(c+\gamma V^{\pi}(0;\gamma)),\nonumber\\ 
V^{\pi}(s;\gamma)= (1-p)(c+wx+\gamma V^{\pi}(s;\gamma))+p(c+\gamma\frac{c}{1-\gamma}),\nonumber\\ 
\therefore \quad V^{\pi}(s;\gamma) = \frac{c+(1-p)wx+\gamma\frac{c}{1-\gamma}}{1-(1-p)\gamma} < \infty. \nonumber
\end{gather}
Hence if $V^{\pi}(s;\gamma)< \infty$, by the definition of optimality $V^{\gamma}(s) < \infty$ 
\end{proof}

\begin{lemma}
\label{cond2}
The optimal discounted cost function  $ V^{\gamma}(x)$ is a non-decreasing function of x. 
\end{lemma}
\begin{proof}
We prove this using recursion on $V_{n}^{\gamma}(s)$. It is easy to see that  $V_{0}^{\gamma}(s)$ is non-decreasing. We assume that $V_{n}^{\gamma}(s)$ is non decreasing.

\begin{gather}
  V^{\gamma}_{n+1}(x;\gamma,a=0) = wp+wx +\gamma(pV^{\gamma}_{n}(x+1)+(1-p)V^{\gamma}_{n}(x)).\\
V^{\gamma}_{n+1}(x;\gamma,a=1) = c+(1-p)wx +\gamma(pV^{\gamma}_{n}(0)+(1-p)V^{\gamma}_{n}(x)).  
\end{gather}
\begin{align*}
V^{\gamma}_{n+1}(x) &= 
\min\left(V^{\gamma}_{n+1}(x;\gamma,a=0),V^{\gamma}_{n+1}(x;\gamma,a=1)\right).\\
&= \min(wp+wx+\gamma pV^{\gamma}_{n}(x+1),c+(1-p)wx+ \gamma pV^{\gamma}_{n}(0))+(1-p)V_{n}^{\gamma}(x)
\end{align*}

 $wp+wx$ and $c+(1-p)wx$ are both non decreasing functions. Also by the induction hypothesis $V^{\gamma}_{n}(x+1)$,  $V^{\gamma}_{n}(x)$ are non decreasing. Since the minimum operation preserves the non-decreasing property  $V^{\gamma}_{n+1}(x)$ is also non-decreasing. 
\end{proof}

\begin{lemma}
\label{avg_optimality}
There exists a deterministic stationary policy  of a MDP $\Delta$ that is cost-optimal if the following three conditions hold \cite{10.2307/171262}:
\begin{enumerate}
    \item For every $\gamma$ and state $s$, $V^{\gamma}(s)$ is finite.  
    \item There exists a non-negative $L$ such that the relative cost function $h^{\gamma} (s) = V^{\gamma}(s) - V^{\gamma}(0) \geq -L \quad \forall \gamma, \forall s$ where 0 is a reference state.
    \item There exists non-negative $M_{s}$, such that $h^{\gamma}(s)\leq M_{s}$ for every $\gamma$ and $s$. For every $s$ there exists an action $a(s)$ such that $\sum_{s'} P(s',a(s),s) M_{s'} < \infty$. Here $P(s',a(s),s)$ represents the transition probability to state $s'$ after taking an action $a(s)$ in state $s$.
\end{enumerate}
\end{lemma}
We prove that these conditions hold true for our MDP.
\begin{proof}
Condition 1 is verified in the proof of Lemma \ref{finite}. Condition 2 can be inferred from the result of Lemma \ref{cond2} with $L$ = 0. We use the optimality equation (\ref{opt}) to prove Condition 3. 
\begin{align*}
V^{\gamma}(x) &= \min\left(V^{\gamma}(x;\gamma,a=0),V^{\gamma}(x;\gamma,a=1)\right).\\
&= \min( wp+wx +\gamma(pV^{\gamma}(x+1)+(1-p)V^{\gamma}(x)),\\& \pushright{c+(1-p)wx +\gamma(pV^{\gamma}(0)+(1-p)V^{\gamma}(x))}).
\end{align*}
\begin{gather*}
   \implies V^{\gamma}(x) \leq c+(1-p)wx +(pV^{\gamma}(0)+(1-p)V^{\gamma}(x)).\\
\implies p( V^{\gamma}(x) - V^{\gamma}(0)) \leq c+ (1-p)wx .\\
\implies h^{\gamma}(x) \leq \frac{c+ (1-p)wx}{p} = M_{x}. 
\end{gather*}

Choosing $a(s) = 1$ proves Condition 3. Hence there exists a deterministic, stationary cost-optimal policy for our MDP. 
\end{proof}
\begin{definition}(Threshold-type policy)
A threshold type policy is a deterministic policy where if it is optimal to schedule in state X, then it is also optimal to schedule in state X+1. 
\end{definition}
\begin{lemma}
The deterministic, stationary cost optimal policy is of threshold type $\forall c \geq 0$.
\end{lemma}
\begin{proof}
Suppose that the $\gamma$-optimal policy for a state s=x is to schedule the sensor i.e.\\ $V^{\gamma}(x;\gamma,a=1)-V^{\gamma}(x;\gamma,a=0) < 0$. Then,\\
$
V^{\gamma}(x+1;\gamma,a=1)-V^{\gamma}(x+1;\gamma,a=0)= \\
 c+(1-p)w(x+1) +\gamma(pV^{\gamma}(0)+(1-p)V^{\gamma}(x+1)) \\
 -( wp+w(x+1) +\gamma(pV^{\gamma}(x+2)+(1-p)V^{\gamma}(x+1)).\\
 \labelrel\leq{seq:increasing} c+(1-p)wx +\gamma(pV^{\gamma}(0)+(1-p)V^{\gamma}(x)) - 
 ( wp+wx +\gamma(pV^{\gamma}(x+1)+(1-p)V^{\gamma}(x)) \\
 = V^{\gamma}(x;\gamma,a=1)-V^{\gamma}(x;\gamma,a=0) < 0.
$\\
We use the non decreasing property of $V^{\gamma}(x)$ in (\ref{seq:increasing}). Hence if it is optimal to schedule the sensor in state $x$, it is optimal to schedule the sensor in state $x+1$. So far we show that $\gamma$-optimal policy is of threshold type. Let $\{\gamma_{n}\}$ be a sequence of discount factors. According to \cite{10.2307/171262} if the MDP satisfies the conditions in Lemma \ref{avg_optimality} then there exists a sub-sequence $\{\beta_{n}\}$, such that the cost-optimal algorithm is the limit point of $\beta_{n}$-optimal algorithms. Using an argument similar to [\cite{6895314}, Theorem 18] the cost-optimal algorithm is also threshold type .\end{proof}
Now we characterize the average-cost of the MDP under the threshold type algorithm. 
\begin{lemma}
Given a threshold type policy with a threshold at $X \in \{1,2,3 \dots\}$, the average cost $C(X,c)$ of the policy is :
\begin{gather}
    C(X,c) = \frac{wX}{2}+\frac{c(2-p)}{X+1}.
\end{gather}
\end{lemma}
\begin{proof}
We evaluate the average cost using the post-action age similar to \cite{article}. The post action age under the threshold type policy forms a discrete time Markov chain (DTMC) as shown in Fig. \ref{fig:dtmc}. The steady state distribution of the DTMC $\left[\pi_{0}, \pi_{1} \dots \pi_{X}\right]$ is:
\begin{gather}
 \pi_{i} = \frac{1}{X+1}   \quad \forall i \in \{ 0,1,\dots X\}. \label{eq:ssd}
\end{gather}
The DTMC associates a cost $C(s)$ with each state $s$. The cost incurred on reaching state 0 is $c$. A cost of $wi$ is incurred for states $i \in \{1,2 \dots X\}$ and an additional cost of $c$ is incurred at state $X$  with probability $(1-p)\pi_{X}$. The additional cost is incurred when the channel is OFF in state X. 

\begin{align*}
\sum_{i=0}^{X} \pi_{i}C(i) &= \frac{c}{X+1} +\frac{\sum_{i=1}^{X} wi}{X+1}+\frac{c(1-p)}{X+1},\\
&=   \frac{wX}{2}+\frac{c(2-p)}{X+1}.
\end{align*}

\end{proof}

We now prove that the subproblem is indexable. Let $a^{*}$ denote the optimal action at state $s=x$. By definition:\\
\begin{gather*}
V^{\gamma}(s,a=0) \underset{a^{*}=1}{\overset{a^{*}=0}{\lessgtr}} V^{\gamma}(s,a=1), \\
wp+wx +\gamma pV^{\gamma}(x+1)\underset{a^{*}=1}{\overset{a^{*}=0}{\lessgtr}} c+(1-p)wx +\gamma pV^{\gamma}(0),\\
 wp+\gamma pV^{\gamma}(x+1)\underset{a^{*}=1}{\overset{a^{*}=0}{\lessgtr}} c-wpx +\gamma pV^{\gamma}(0).
 \end{gather*}
 
$V^{\gamma}(x)$ is non-decreasing in $x$. Hence if $c\leq 0$ then $a^{*} = 1$ for all states i.e $S(0) = \phi$. If $c\geq 0 \: \exists  x_{0} $ such that:$ wp+ \gamma pV^{\gamma}(x+1) > c-wpx +\gamma pV^{\gamma}(0) \quad \forall x \geq x_{0}.$\\
Hence $a^{*}=0$ upto state $x_{0}$. As $c$ increases, so does $x_{0}$ and S(c) monotonically increases from the empty set to the entire state space as $c\rightarrow \infty$. Therefore the subproblem is indexable. Note that $C(X,c)$ is convex in X and let $X^{*}(c)$ (not necessarily an integer) minimize $C(X,c)$. Hence the optimal threshold is $\floor{X^{*}(c)}$ or $\ceil{X^{*}(c)}$.  Recall that the Whittle index of a state is the playing charge $c$ that makes it equally optimal to play or idle the arm. This means that $C(\floor{X^{*}(c)},c)=C(\ceil{X^{*}(c)},c) $. In other words the Whittle index of the state x is the playing charge c that satisfies $C(x,c) = C(x+1,c)$. This completes the proof. 
\begin{figure}[h]
\centering
    \includegraphics[scale = 0.60]{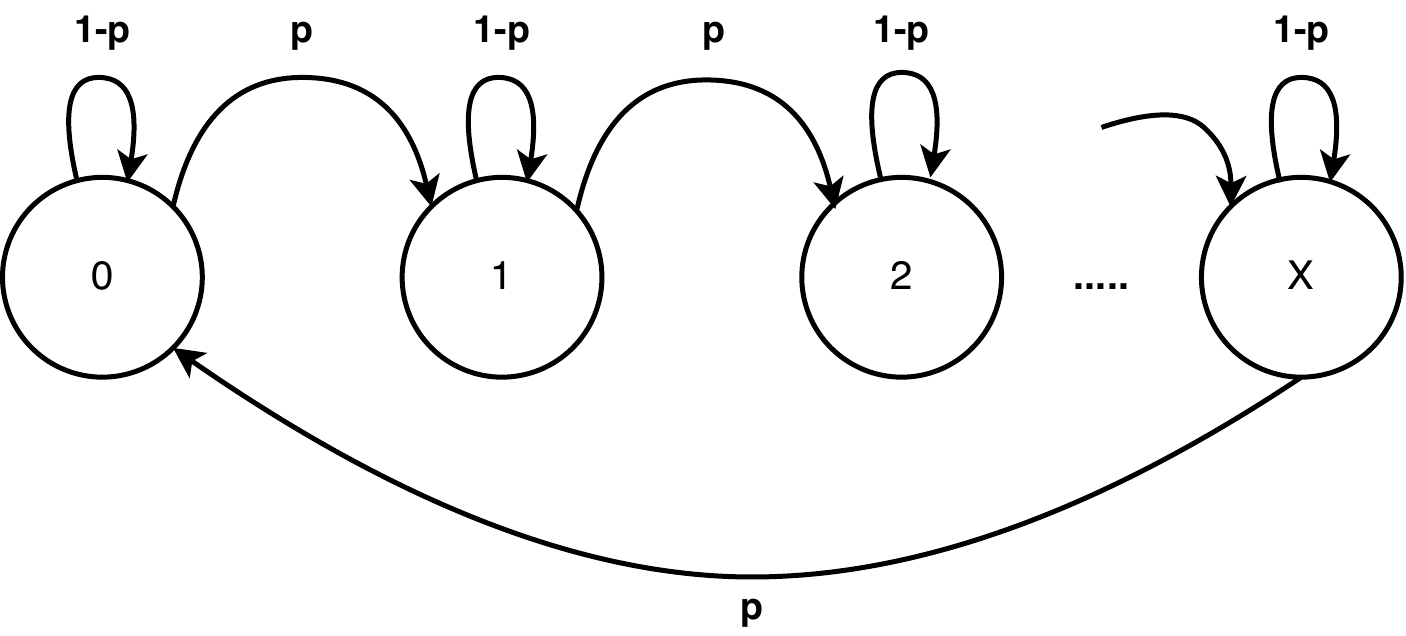}
    \caption{DTMC of the post action age under a threshold type policy }
    \label{fig:dtmc}
\end{figure}

\subsection{Proof of Theorem \ref{theorem:CSI_whittle}}
\label{appe:Th2}
We follow the technique used in the previous section. 

\begin{lemma}
\label{CSI_cond1}
We define an iteration on  $V^{\gamma}_{n}(s)$ with  $V^{\gamma}_{0}(s) = 0 \; \forall s$
\begin{gather}
V^{\gamma}_{n+1}(s) = \min_{a\in\{0,1\}} c(s,a) + \gamma \mathbb{E}[V^{\gamma}_{n}(s')] \label{value_iterm_2}
\end{gather}
Then  $V^{\gamma}_{n}(s)\rightarrow V^{\gamma}(s)$ as $n\rightarrow\infty$
\end{lemma}
\begin{proof}
We show that $V^{\gamma}(s) < \infty$. Let $\pi$ be the policy of scheduling the sensor in every state. Then under $\pi$ the discounted cost function of state $(x,\Lambda)$ is:
\begin{gather*}
V^{\pi}((x,1);\gamma) = \frac{ c}{1-\gamma}  < \infty,\\
V^{\pi}((x,0);\gamma) = \frac{wx}{1-(1-p)\gamma}+\frac{\gamma p c}{(1-\gamma)(1-(1-p)\gamma)}+\frac{c}{1-(1-p)\gamma} < \infty.
\end{gather*}

Hence $v^{\gamma}(s)< \infty \: \forall s$.
\end{proof}

\begin{lemma}
\label{CSI_cond2}
The optimal discounted cost function $V^{\gamma}((x,\Lambda))$ is a non-decreasing function of x for a fixed channel state $\Lambda$. 
\end{lemma}
\begin{proof}
Like Lemma \ref{cond2} we can prove this by recursion on $V^{\gamma}_{n}(s)$
\end{proof}
\begin{lemma}
There exists a deterministic, stationary policy for the MDP $\Delta$ of the sub-problem that is cost-optimal 
\end{lemma}
\begin{proof}
We verify the three conditions in Lemma \ref{avg_optimality}. Condition 1 holds true through the proof of Lemma \ref{CSI_cond1}. Note that for $\forall c\geq0, \: V^{\gamma}((0,1))\geq V^\gamma((0,0))$. Hence using this result along with the result of Lemma \ref{CSI_cond2} the relative cost function $h^{\gamma}(s) = V^{\gamma}(s) - V^{\gamma}((0,0)) \geq 0$. Using the optimality equation (\ref{opt}) for $V^{\gamma}((x,0))$ and $V^{\gamma}((x,1))$ one can prove:
\begin{gather}
  h^{\gamma}((x,\Lambda)) \leq \frac{wx(1-\Lambda)}{p}+\frac{c}{1-p} = M_{(x,\Lambda)}.  \nonumber
\end{gather}

Using $a(s)=1$ verifies condition 3. 
\end{proof}

\begin{lemma}
The stationary, deterministic cost-optimal policy is threshold type $\forall c \geq 0$ 
\end{lemma}
\begin{proof}
We first show the threshold structure of the $\gamma$-optimal policy. When $c \geq 0$ it is obvious that the optimal action for all states $(x,0)$ is to idle. Suppose it is optimal to schedule the sensor in state $(x,1)$ i.e. $V^{\gamma}((x,1);\gamma,a=1)-V^{\gamma}((x,1);\gamma,a=0) < 0$ Then,\\
$V^{\gamma}((x+1,1);\gamma,a=1)-V^{\gamma}((x+1,1);\gamma,a=0) =\\ c+\gamma p V^{\gamma}(0,1)+\gamma (1-p) V^{\gamma}(0,0) - w(x+2) - \gamma p V^{\gamma}(x+2,1)+\gamma (1-p) V^{\gamma}(x+2,0)
\\ \labelrel\leq{seq:non_decreasing_csi} c+\gamma p V^{\gamma}(0,1)+\gamma (1-p) V^{\gamma}(0,0) - w(x+1) - \gamma p V^{\gamma}(x+1,1)+\gamma (1-p) V^{\gamma}(x+1,0) \\=
V^{\gamma}((x,1);\gamma,a=1)-V^{\gamma}((x,1);\gamma,a=0) < 0.$

We use the non decreasing property of $V^{\gamma}((x,1))$ in (\ref{seq:non_decreasing_csi}). Hence if it is optimal to schedule the sensor in state $(x,1)$, it is optimal to schedule the sensor in state $(x+1,1)$. Since the MDP satisfies the conditions in Lemma \ref{avg_optimality}, the cost-optimal algorithm is also threshold type.
\end{proof}
We characterize the average cost of the threshold type policy next.

\begin{lemma}
Given a threshold type policy with a threshold at $X \in \{1,2,3 \dots\}$, the average cost $C(X,c)$ of the policy is :
\begin{gather}
    C(X,c) = \frac{wX}{2}+\frac{c}{X+1}.
\end{gather}
\end{lemma}
\begin{proof}
We again use the post-action age DTMC to evaluate the average cost. The DTMC induced by the policy thresholded at $X$ is shown in Fig \ref{fig:dtmc}. The steady state distribution of the DTMC $\left[\pi_{0}, \pi_{1} \dots \pi_{X}\right]$ is given by (\ref{eq:ssd}). 
The cost $C(i)$ associated with state 0 is c and the cost incurred in state $i \in \{2,3 \dots X\}$ is $wi$. Then,
$ \sum_{i=0}^{X} \pi_{i}C(i) = \frac{wX}{2}+\frac{c}{X+1}$
\end{proof}
Next we prove that the sub-problem is indexable. Using the optimality equation it is obvious that for all states $(x,0)$ it is optimal to schedule if $c \leq 0$ and optimal to idle if $c>0$. Let $a^{*}$ be the optimal action in state $(x,1)$. Now by definition, 
\begin{gather*}
   V^{\gamma}((x,1),a=0) \underset{a^{*}=1}{\overset{a^{*}=0}{\lessgtr}} V^{\gamma}((x,1),a=1).\\  w(x+1)+\gamma \left( p V^{\gamma}(x+1,1)+ (1-p) V^{\gamma}(x+1,0) \right) \underset{a^{*}=1}{\overset{a^{*}=0}{\lessgtr}}
c+\gamma p V^{\gamma}(0,1)+\gamma (1-p) V^{\gamma}(0,0). 
\end{gather*}

For $c\leq 0$ it is optimal to schedule in every state and therefore $S(c) = \phi$. If $c \geq 0 \: \exists x_{0}$ such that:\\

\begin{align*}
w(x+1)+\gamma \left( p V^{\gamma}(x+1,1)+ (1-p) V^{\gamma}(x+1,0) \right)& > c+\gamma p V^{\gamma}(0,1)+\gamma (1-p) V^{\gamma}(0,0) \\&\pushright{\forall x \geq x_{0}.}
\end{align*}

Hence $a^{*} = 0 \quad \forall x \leq x_{0}$ . As $c$ increases, so does $x_{0}$ and S(c) monotonically increases from the empty set to the entire state space as $c\rightarrow \infty$. Therefore the subproblem is indexable. $C(X,c)$ is convex in $X$ and using the argument given in the previous section the Whittle index is given by the playing charge $c$ that satisfies $C(x,c) = C(x+1,c)$. This completes the proof.

\subsection{Proof of Theorem \ref{theorem:UCSI_random}}
\label{appe:Th4}
To prove Theorem \ref{theorem:UCSI_random} we first analyse $\mathbb{E}_{\pi^{-}}\left[X_{i}(t)\right]$.
\begin{lemma}
If $p_{i}>0$ and $\Delta_{i}>0$,\\
\[\mathbb{E}_{\pi^{-}}\left[X_{i}(t)\right] = \frac{1-\Delta_{i}}{\Delta_{i}}\left(1 - \delta_{i}^{t}\right). \]\\
Where $\delta_{i} = 1 - p_{i} \Delta_{i}.$
\end{lemma}
\begin{proof}
 Let the random variable $D_{i}$ indicate the time elapsed since the last successful update by sensor $i$\\\\
 $\mathbb{E}_{\pi^{-}}\left[X_{i}(t)\right]=\sum_{d=0}^{t}\mathbb{E}_{\pi^{-}}\left[X_{i}(t)|D_{i}=d\right]\mathbb{P}(D_{i}=d).$\\
Let $t_{prev}$ be the time stamp of the last successful update. Then,\\
\begin{equation}
  X_{i}(t) = \sum_{m=0}^{t -t_{prev}-1} \mathbbm{1}_{\{\Lambda_{i}(t-m)=1 \text{ and } a_{i}(t-m) = 0\}}. \nonumber  
\end{equation}
Where $\mathbbm{1}$ is the indicator function. Therefore,
\begin{align*}
 \mathbb{E}_{\pi^{-}}\left[X_{i}(t)|D_{i}=d\right]
 &= \sum_{m=0}^{d-1} \mathbbm{E}\left[ \mathbbm{1}_{\{\Lambda_{i}(t-m)=1 \text{ and } a_{i}(t-m) = 0\}} |D_{i}=d\right],\\
&= d \times \mathbb{P}(\Lambda_{i}(t-m)=1 \text{ and } a_{i}(t-m) = 0 |D_{i}=d),\\
&= \frac{dp_{i} (1-\Delta_{i})}{1-p_{i}\Delta_{i}}\\ &= \frac{dp_{i} (1-\Delta_{i})}{\delta_{i}}.
\end{align*}

$\mathbb{P}(D_{i}=d) = \begin{cases}
p_{i}\Delta_{i}\delta_{i}^{d}& \text{ if }d < t\\
\delta_{i}^{t}& \text{Otherwise}
\end{cases}$\\\\
Hence,
\begin{align*}
    \mathbb{E}_{\pi^{-}}\left[X_{i}(t)\right] &= \frac{\Delta_{i}(1-\Delta_{i})p_{i}^{2}}{\delta_{i}}\sum_{d=0}^{t-1} d\delta_{i}^{d}+\frac{\delta^{t}_{i}tp_{i}(1-\Delta_{i})}{\delta_{i}},\\
    &= \frac{\Delta_{i}(1-\Delta_{i})p_{i}^{2}}{(1-\delta_{i})^{2}}\left(1 - \delta_{i}^{t} \right). 
\end{align*}

Substituting $\delta_{i} = 1-p_{i}\Delta_{i}$ in the denominator completes the proof.
\end{proof}
Since we assume $\vec{p}>0$,
\begin{align*}
 J^{\pi^{-}} &= \limsup\limits_{T\rightarrow \infty} \frac{1}{T} \mathbb{E}_{\pi^{-}} \left[\sum_{t=1}^{T}\sum_{i\in n^{-}} w_{i}X_{i}(t)\right],\\
 &= \limsup\limits_{T\rightarrow \infty} \frac{1}{T}\left(\sum_{i \in n^{-}}w_{i}\sum_{t=1}^{T} \frac{1-\Delta_{i}}{\Delta_{i}}\left(1 - \delta_{i}^{t}\right)\right),\\
 &= \sum_{i \in n^{-}}w_{i}\frac{(1-\Delta_{i})}{\Delta_{i}}.
\end{align*}

\subsection{Analysis of Algorithm \ref{algo_1}}
\label{appe:alg1}
Note that we omit the inequality constraints $1\geq\vec{\Delta}>0$ while setting up (\ref{eq_deltas}) and (\ref{eq_lamda}). However it is obvious that the real solution $\lambda^{*}>0$ and so  $\vec{\Delta}^{*}>0$. Moreover (\ref{eq_lamda}) ensures that $\vec{\Delta}^{*} \leq 1$.
\subsection{Proof of Theorem \ref{theorem:CSI_random}}
\label{appe:Th5}
We first analyze  $\mathbb{E}_{\pi^{+}}\left[X_{i}(t)\right]$.
\begin{lemma}
If $p_{i}>0$ and $\alpha_{i}>0$,
\[\mathbb{E}_{\pi^{+}}\left[X_{i}(t)\right] = \frac{\beta_{i}}{\alpha_{i}}\left(1-(p_{i}\beta_{i}+q_{i})^{t}\right).\]\\
Where $\beta_{i} = 1-\alpha_{i},\: q_{i} = 1-p_{i}$
\end{lemma}
\begin{proof}
Let the random variable $C_{i}(t)$ be the number of times the channel is in the ON-state uptil time $t$.  \\
$\mathbb{E}_{\pi^{+}}\left[X_{i}(t)\right]=\sum_{c=0}^{t}\mathbb{E}_{\pi^{+}}\left[X_{i}(t)|C_{i}(t)=c\right]\mathbb{P}(C_{i}(t)=c).$\\
Now,\\
$\mathbb{P}(X_{i}(t)=m|C_{i}(t)=c)=\begin{cases}
\alpha_{i}\beta_{i}^{m} & \text{ if } m < c\\
\beta_{i}^{c} & \text{Otherwise}
\end{cases}$\\
Hence,
\begin{align*}
\mathbb{E}_{\pi^{+}}\left[X_{i}(t)|C_{i}(t)=c\right] &= \sum_{m=0}^{c}m\mathbb{P}(X_{i}(t)=m|C_{i}(t)=c),\\
&= \alpha_{i}\sum_{m=0}^{c-1}m\beta_{i}^{m}+\beta_{i}^{c}c,\\
&= \frac{\beta_{i}-\beta_{i}^{c+1}}{\alpha_{i}}.
\end{align*}
Also,
\begin{gather*}
\mathbb{P}(C_{i}(t)=c)= \binom{t}{c} p_{i}^{c}(1-p_{i})^{t-c}.
\end{gather*}
Hence,
\begin{align*}
\mathbb{E}_{\pi^{+}}\left[X_{i}(t)\right] &= \sum_{c=0}^{t}\left(\frac{\beta_{i}-\beta_{i}^{c+1}}{\alpha_{i}}\right)\times \binom{t}{c} p_{i}^{c}(1-p_{i})^{t-c,}\\
&= \frac{\beta_{i}}{\alpha_{i}}\left(1-\sum_{c=0}^{t} \binom{t}{c} (p_{i}\beta_{i})^{c}(1-p_{i})^{t-c} \right),\\
&= \frac{\beta_{i}}{\alpha_{i}}\left(1-(p_{i}\beta_{i}+q_{i})^{t}\right).
\end{align*}

\end{proof}
Since we assume $\vec{p}>0$,
\begin{align*}
 J^{\pi^{+}} &= \limsup\limits_{T\rightarrow \infty} \frac{1}{T} \mathbb{E}_{\pi} \left[\sum_{t=1}^{T}\sum_{i\in n^{+}} w_{i}X_{i}(t)\right],\\
 &= \limsup\limits_{T\rightarrow \infty} \frac{1}{T}\left(\sum_{i\in n^{+}}w_{i}\sum_{t=1}^{T}\frac{\beta_{i}}{\alpha_{i}}(1-(p_{i}\beta_{i}+q_{i})^{t})\right),\\
\end{align*}
$p_{i}\beta_{i}+q_{i} = 1-p_{i}\alpha_{i} < 1$ and hence the second term in the sum is dominated by the $T$ in the denominator while taking the limit. This gives us the result of Theorem \ref{theorem:CSI_random}
\subsection{Analysis of Algorithm \ref{algo_2}}
\label{appe:alg2}
We do not account for $\vec{\alpha }> 0$ while setting up (\ref{csi_opt1}) and (\ref{eq_csi}). However it is obvious that the solution  $\lambda^{*} > 0$ and so $\vec{\alpha}^{*}>0$. The algorithm design ensures that  $\vec{\alpha}^{*}\leq 1$

\end{document}